\documentclass[11pt]{article}

\usepackage{url}
\usepackage{color,microtype}
\usepackage{graphicx,subfigure}
\usepackage{amssymb,amsmath,amssymb,amsfonts}
\usepackage{amsthm}
\usepackage{textcomp}
\usepackage{enumerate}
\usepackage{multirow}
\usepackage{gensymb}
\usepackage{array}
\usepackage{algorithm}
\usepackage[noend]{algorithmic}
\usepackage{graphicx}
\usepackage{comment}
\usepackage{mathptmx}
\DeclareMathAlphabet{\mathcal}{OMS}{cmsy}{m}{n}

\usepackage[sort]{cite}
\newcommand{\argmax}{\operatornamewithlimits{argmax}}

\newcommand{\vol}{{\ensuremath{\mathit{vol}}}}
\newcommand{\cut}{{\ensuremath{\mathit{cut}}}}

\newcounter{lp}
\newcommand{\lptag}{\tag{LP\arabic{lp}}\addtocounter{lp}{1}}
\newcounter{sdp}
\newcommand{\cgeq}{1}
\newcommand{\cleq}{2}
\newcommand{\cneg}{3}

\newcommand{\prob}[1]{\textrm{Pr} \left (#1 \right )}

\newcommand{\calC}{\mathcal{C}}
\newcommand{\bA}{\mathbf{A}}

\newcommand{\bC}{\mathbf{C}}
\newcommand{\bP}{\mathbf{P}}
\newcommand{\bD}{\mathbf{D}}
\newcommand{\bX}{\mathbf{X}}
\newcommand{\bY}{\mathbf{Y}}

\newcommand{\bF}{\mathbf{F}}

\newcommand{\bZ}{\mathbf{Z}}

\newcommand{\bb}{\mathbf{b}}
\newcommand{\bc}{\mathbf{c}}
\newcommand{\bx}{\mathbf{x}}
\newcommand{\by}{\mathbf{y}}

\newcommand{\bu}{\mathbf{u}}
\newcommand{\bg}{\mathbf{g}}

\newcommand{\mwmlp}{{MWM-LP}}

\newtheorem{definition}{Definition}

\newtheorem{theorem}{Theorem}

\newtheorem{corollary}[theorem]{Corollary}
\newtheorem{lemma}[theorem]{Lemma}

\mathchardef\mhyphen="2D

\newcommand{\expec}[1]{\mathbb E\left [ #1 \right ]}
\newcommand{\var}[1]{\mathbb V\left [ #1 \right ]}
\newcommand{\rounddown}[1]{\left \lfloor  #1 \right \rfloor}
\newcommand{\roundup}[1]{\left \lceil  #1 \right \rceil}
\newcommand{\mindisagree}{\ensuremath{\mathsf{min\mhyphen disagree}}}
\newcommand{\maxagree}{\ensuremath{\mathsf{max\mhyphen agree}}}
\newcommand{\disagree}{\ensuremath{\mathsf{disagree}}}
\newcommand{\agree}{\ensuremath{\mathsf{agree}}}
\DeclareMathOperator{\polylog}{polylog}
\DeclareMathOperator{\poly}{poly}
\DeclareMathOperator{\argmin}{argmin}
\newcommand{\C}{\ensuremath\mathcal{C}}
\newcommand{\ind}{{\sc Index}}
\newcommand{\disj}{{\sc Disj}}

\usepackage{fullpage}
\usepackage{times}
\usepackage{graphicx} 
\usepackage{subfigure} 
\usepackage{natbib}
\usepackage{algorithm}
\usepackage{algorithmic}

\usepackage{hyperref}

\begin{document} 

\title{Correlation Clustering in Data Streams}
\author{
Kook Jin Ahn\thanks{University of Pennsylvania. \texttt{kookjin@cis.upenn.edu}. The author is currently at Google, \texttt{kookjin@google.com}.}
\and
Graham Cormode\thanks{University of Warwick.  Supported in part by European Research Council grant ERC-2014-CoG 647557, a Royal  Society Wolfson Research Merit Award and the Yahoo Faculty Research  Engagement Program.  \texttt{G.Cormode@warwick.ac.uk}}
\and
Sudipto Guha\thanks{This work was done while the author was at University of Pennsylvania, supported by NSF Award CCF-1546141. \texttt{sudipto@cis.upenn.edu}}  
 \and  
Andrew McGregor\thanks{University of Massachusetts Amherst. Supported by NSF  Award CCF-1637536. \texttt{mcgregor@cs.umass.edu}}
\and
Anthony Wirth\thanks{School of Computing and Information Systems, The University of Melbourne. Supported by ARC Future Fellowship FT120100307. \texttt{awirth@unimelb.edu.au}}}
\maketitle

\begin{abstract}
Clustering is a fundamental tool for analyzing large data sets. A
rich body of work has been devoted to designing data-stream algorithms
for the relevant optimization problems such as $k$-center, $k$-median,
and $k$-means. Such algorithms need to be both time  and and space efficient. 
In this paper, we address the problem of \emph{correlation clustering} in the dynamic data stream
model. The stream consists of updates to the edge weights of a graph
on~$n$ nodes and the goal is to find a node-partition such that the
end-points of negative-weight edges are typically in different
clusters whereas the end-points of positive-weight edges are typically
in the same cluster. We present polynomial-time, $O(n\cdot \polylog
n)$-space approximation algorithms for natural problems that arise.

We first develop data structures based on linear sketches that allow
the ``quality'' of a given node-partition to be measured. We then 
combine these data structures with convex programming and
sampling techniques to solve the relevant approximation problem.
Unfortunately, the standard LP and SDP formulations are not
obviously
solvable in $O(n\cdot \polylog n)$-space. Our work presents
space-efficient algorithms for the convex programming required, as well as
approaches to reduce the adaptivity of the sampling.
\end{abstract}

\newpage

\section{Introduction}

The correlation clustering problem was first formulated as an optimization
problem by~\citet*{BansalBC04}.
The input is a complete weighted graph~$G$ on $n$~nodes, where each pair of
nodes~$uv$
has weight $w_{uv}\in {\mathbb R}$.
A positive-weight edge indicates that~$u$
and~$v$ should be in the same cluster, whereas a negative-weight edge
indicates that~$u$ and~$v$ should be in different clusters.
Given a node-partition $\C=\{C_1,C_2,\ldots \}$, we say edge
$uv$ \emph{agrees} with~$\C$, denoted by $uv \sim \C$, if the relevant soft
constraint is observed.
The goal is to find the partition~$\C$ that maximizes 
\[
\agree(G,\C):=\sum_{uv \sim \C} |w_{uv}| 
\, 
\]
or, equivalently, that minimizes $\disagree(G,\C) := \sum_{uv}
|w_{uv}|-\agree(G,\C)$. 
Solving this problem exactly is known to be NP-hard.
A large body of work
has been devoted to approximating $\maxagree(G)=\max_{\C} \agree(G,\C)$ and
$\mindisagree(G)=\min_{\C} \disagree(G,\C)$, along with variants
$\mindisagree_k(G)$ and $\maxagree_k(G)$, where we consider partitions with at most $k$ clusters. 
In this paper, we focus on multiplicative approximation results. 
If all weights are $\pm 1$, there is a polynomial time approximation scheme
(PTAS) for  $\maxagree$ \citep{BansalBC04,GiotisG06}
and a 
 $2.06$-approximation \citep{ChawlaMSY15},
for $\mindisagree$. 
When there is an upper bound,~$k$, on the number of clusters in~$\C$, and all
weights are~$\pm 1$, \citet{GiotisG06} introduced a PTAS for both problems.
Even $k=2$ is interesting, with an efficient local-search approximation
introduced by~\citet*{ColemanSW08}.

If the weights are arbitrary, there is a $0.7666$-approximation for
$\maxagree$ \citep*{Swamy04,CharikarGW05} and an $O(\log n)$-approximation for
$\mindisagree$ \citep{CharikarGW05,DemaineEFI06}.
These methods use convex programming: as originally described, this
cannot be implemented in $O(n\polylog n)$ space, even when
 the input graph is sparse.
This aspect is well known in practice,
and
\citet*{ES09,BG11,kdd2014} discuss the difficulty of scaling the convex programming approach.

\paragraph{Clustering and Graph Analysis in Data Streams.}
Given the importance of clustering as
a basic tool for analyzing massive data sets, it is unsurprising that a considerable effort
has gone into designing clustering algorithms in the relevant
computational models.
In particular, in the data-stream model we are
permitted a limited number of passes (ideally just one) over the data
while using only limited memory. This model abstracts the challenges
in traditional applications of stream processing such as network
monitoring, and also leads to I/O-efficient external-memory
algorithms. Naturally, in either context, an algorithm should also be fast,
both in terms of the time to process each stream element
and in returning the final answer.

Classical clustering problems including
$k$-median~\citep{GuhaMMO00,CharikarOP03}, $k$-means~\citep*{AilonJM09}, and
$k$-center~\citep{Guha09,mk2008,CharikarCFM04} have all been studied in the
data stream model, as surveyed by~\citet{Silva13}.
Non-adaptive sampling algorithms for correlation clustering can be implemented
in the data stream model, as applied by~\citet{arxiv}, to
construct {\em additive} approximations.
\cite*{ChierichettiDK14} presented the first multiplicative approximation data stream algorithm: a polynomial-time
$(3+\epsilon)$-approximation for $\mindisagree$ on $\pm1$-weighted
graphs using $O(\epsilon^{-1} \log^2 n)$ passes and \emph{semi-streaming} space --- that is,
a streaming algorithm using $\Theta (n \polylog n)$
memory~\citep{FKMSZ05b}.
\citet{PanPORRJ15} and \citet{kdd2014} discuss faster non-streaming implementations of related ideas but  \cite*{ChierichettiDK14} remained the state of the art data stream algorithm until our work. Using space roughly proportional to the number of nodes can be shown to be necessary for solving many natural graph problems including, it will turn out, correlation clustering. For a recent survey of the semi-streaming algorithms and graph sketching see \citet{McGregor14}. 

\paragraph{Computational Model.}
In the basic graph stream model, the input is a sequence of edges and their weights.
The available space to process the stream and perform any necessary post-processing is $O(n \polylog n)$ bits.
Our results also extend to the \emph{dynamic graph stream model} where the stream consists of both insertions and deletions of edges; the weight of an edge is specified when the edge is inserted and deleted (if it is subsequently deleted).
For simplicity, we assume that all weights are integral. We will consider three types of weighted graphs: (a) \emph{unit weights}, where
all $w_{uv}\in \{-1,1\}$; (b) \emph{bounded weights}, where all weights
are in the range $[-w_*,-1]\cup [1,w_*]$ for some constant $w_*\geq 1$;
and (c) \emph{arbitrary weights}, where  all weights
are in the range $[-w_*,w_*]$ where $w_*=\poly(n)$. 
We denote the sets of positive-weight and negative-weight edges by~$E^+$
and~$E^-$, respectively,  and define $G^+=(V,E^+)$ and $G^-=(V,E^-)$. 

We note that many of our algorithms, such as those based on sparsification~\citep{access}, can also be implemented in MapReduce.

\subsection{Our Results}
We summarize our results in Table~\ref{table:summary}.

\paragraph{Max-Agree.}
For $\maxagree$, we provide the following single-pass streaming
algorithms, each needing $\tilde{O}(n\epsilon^{-2})$ space: (i) a
polynomial-time $(1-\epsilon)$-approximation for bounded weights
(Theorem~\ref{thm:ptasboundedagree}),
and (ii) a $0.766(1-\epsilon)$ approximation for arbitrary weights in
$\tilde{O}(n\epsilon^{-10})$ time
(Theorem~\ref{thm:max}).

\begin{table}
  \caption{Summary of approximation results in this paper.}
  \label{table:summary}
  \centering
  \begin{tabular}{|c|c|c|c|c|c|}
\hline
    Section & Problem & Weights & Passes & Space Bound & Approximation Factor \\
    \hline
    \ref{sec:l2combrec} & \disagree & unit & 1 & $\tilde{O}(\epsilon^{-2})$ & $1+\epsilon$  \\
    \ref{sec:l2combrec} & \mindisagree & unit & 1 & $O(n + \epsilon^{-2} t)$ & $1+\epsilon$ if $\mindisagree(G)\leq t$  \\
    \ref{sec:2ds} & \maxagree & bounded & 1 & $O(n \poly(k,
    \epsilon^{-1}))$ & $1-\epsilon$   \\
    \ref{sec:3ds} & $\disagree_2$ & arbitrary & 1 &
    $\tilde{O}(n\epsilon^{-2})$ & $1+\epsilon$   \\
    \ref{sec:3ds} & $\mindisagree_2$ & bounded & 1 & $\tilde{O}(n
    \epsilon^{-2})$ & $1+\epsilon$  \\
    \ref{sec:mindisagree} &  \mindisagree & arbitrary & 1 &
    $\tilde{O}(n\epsilon^{-2} + |E^{-}|)$  & $(3 + \epsilon) \log |E|$  \\
    \ref{sec:sdp} & \maxagree & arbitrary & 1 & 
    $\tilde{O}(n\epsilon^{-2})$ & $0.7666-\epsilon$  \\
    \hline
    \ref{sec:mind3} & \mindisagree & unit & $\log \log n$ &
    $\tilde{O}(n)$ & 3\\
    \ref{sec:mdfc} & $\mindisagree_k$ & unit & $\log \log n$ &
    $\tilde{O}(n \poly(k,\epsilon^{-1}))$ & $1+\epsilon$  \\
\hline
\ref{sec:lower} & \mindisagree & unit & $p$ & $\Omega(n/p)$ & Any  \\
\ref{sec:lower} & \mindisagree & arbitrary & 1 & $\Omega(n+|E^-|)$ &  Any   \\
\ref{sec:lower} & $\disagree_k$  for $k\geq 3$ & arbitrary & 1 & $\Omega(n^2)$ & Any \\
    \hline
  \end{tabular}
\end{table}

\paragraph{Min-Disagree.}

We  show that any constant pass algorithm that can test whether $\mindisagree(G)=0$
in a single pass,
for unit  weights, must store $\Omega(n)$ bits 
 (Theorems~\ref{thm:n3lb}).
For arbitrary weights, the lower bound increases to $\Omega(n+|E^-|)$ (Theorem~\ref{thm:n2lb})
and  to  $\Omega(n^2)$ in the case the graph of negative edges may be dense.
We provide a single-pass algorithm that uses $s=\tilde{O}(n \epsilon^{-2}
 + |E^-|)$ space
and $\tilde{O}(s^2)$
time and provides an $O(\log |E^-|)$ approximation (Theorem~\ref{thm:min}).
Since \citet{DemaineEFI06} and \citet{CharikarGW05} provide
approximation-preserving reductions from the ``minimum multicut" problem to
$\mindisagree$ with arbitrary weights,
it is expected to be difficult
to approximate the latter to 
 better than a $\log |E^-|$ factor in polynomial time.
For unit weights when $\mindisagree(G) \leq t$, we provide a single-pass
polynomial time algorithm that uses $\tilde{O}(n + t)$ space
(Theorem~\ref{thm:smalldisagree}). We provide a
$\tilde{O}(n \epsilon^{-2} )$-space PTAS for $\mindisagree_2$
for bounded weights (Theorem~\ref{thm:boundeddisagree2}). 

We also consider multiple-pass streaming algorithms.
For unit weights,
we present a $O(\log \log n)$-pass algorithm that mimics the algorithm
of~\citet{AilonCN08}, and provides
a $3$-approximation in expectation (Theorem~\ref{thm:mult1}),
improving on the result of~\citet{ChierichettiDK14}.
For $\mindisagree_k(G)$, on unit-weight graphs with $k \geq 3$, we give
a $\min(k-1,O(\log \log n))$-pass polynomial-time algorithm 
using $\tilde{O}(n\epsilon^{-2})$ space (Theorem~\ref{thm:mult2}). This result is based on emulating an algorithm by \cite{GiotisG06} in the data stream model.

\subsection{Techniques and Roadmap}
In Section~\ref{sec:basic}, we present three basic data structures for the $\agree$
and $\disagree$ query problems where a partition~$\C$ is specified at
the end of the stream, and the goal is to return an approximation of
$\agree(G,\C)$ or $\disagree(G,\C)$. 
They are based on linear sketches and incorporate
ideas from work on constructing graph sparsifiers via linear
sketches.
These data structures can be
constructed in the semi-streaming model and can be queried in
$\tilde{O}(n)$ time. 
As algorithms rely on relatively simple matrix-vector operations, they
can be implemented fairly easily in MapReduce. 

In Section~\ref{sec:convexsmall} and \ref{sec:sdp}, we introduce several
new ideas for solving the LP and SDP for 
\mindisagree\ and \maxagree.
In each case, the convex formulation must allow each candidate solution
to be represented, verified, and updated in small space. 
But the key point made here is that the formulation plays an outsized role
in terms of space efficiency, both from the perspective of the state required to compute
and the operational perspective of efficiently updating that state.
In future, we expect the space efficiency of solving convex optimization to be
increasingly important.

We discuss multipass for algorithms for \mindisagree\  in Section~\ref{sec:multipass}. Our results are based on adapting existing algorithms that, if implemented in the data stream model, may appear to take $O(n)$ passes. However, with a more careful analysis we show that $O(\log \log n)$ passes are sufficient. Finally, we present space lower bounds  in Section~\ref{sec:lower}. These are proved using reductions from communication complexity and establish that many of our algorithms are space-optimal.

\section{Basic Data Structures and Applications}
\label{sec:basic}

We introduce three basic data structures that can be constructed with a
single-pass over the input stream that defines the weighted graph~$G$.
Given a query partition~$\calC$, these data structures return estimates of
$\agree(G,\calC)$ or $\disagree(G,\calC)$. 
Solving the correlation clustering optimization problem with these structures
directly would require exponential time or $\omega(n \polylog n)$ space.
Instead, we will exploit them carefully to design more efficient
solutions. However, in this section, we will present a short application of each data structure
that illustrates their utility.

\subsection{First Data Structure: Bilinear Sketch}\label{sec:l2combrec}

Consider a graph~$G$ with unit weights ($w_{ij}\in \{-1,1\}$)
and a clustering~$\calC$.
Our first data structure allows us to solve the {\em query problem},
which is, given $G$ and $\calC$, to report (an approximation of) $\disagree(G,\calC)$. 
Define the matrices~$M^G$
and~$M^\calC$ where $M_{ij}^G=\max(0,w_{ij})$ and 
\[M_{ij}^\calC=
\begin{cases}
0 & \mbox{ if $i$ and $j$ are separated in $\calC$} \\
1 & \mbox{ if $i$ and $j$ are not separated in $\calC$} \,.
\end{cases}
\]
Note that if $w_{ij}=1$, then 
\[(M_{ij}^G-M_{ij}^\calC)^2=(1-M_{ij}^\calC)^2=I[\mbox{$i$ and $j$ are separated in $\calC$}]\]
whereas, if $w_{ij}=-1$ then 
\[(M_{ij}^G-M_{ij}^\calC)^2=(M_{ij}^\calC)^2=I[\mbox{$i$ and $j$ are not separated in $\calC$}] \ . \]

Hence, the (squared) matrix distance, induced by the Frobenius norm,
gives exactly
\[\disagree(G,\calC)=\|M^G-M^\calC\|_F^2=\sum_{ij} (M_{ij}^G-M_{ij}^\calC)^2 \ .\]
To efficiently estimate $\|M^G-M^\calC\|_F^2$ when $\calC$ is not known a
priori, we can repurpose the bilinear sketch approach of \citet{IndykM08}. The basic sketch is as follows:
\begin{enumerate}
\item Let $\alpha\in \{-1,1\}^n$ and $\beta\in \{-1,1\}^n$ be independent
random vectors whose entries are 4-wise independent;
in a single pass over the input,
compute \[Y=\sum_{ij\in E^+} \alpha_i \beta_j \ .\]
 Specifically, we maintain a counter that is initialized to 0 and for each $ij\in E^+$ in the stream we add $\alpha_i \beta_j$ to the counter and if $ij\in E^+$ is deleted we subtract $\alpha_i \beta_j$ from the counter; the final value of the counter equals $Y$. Note that $\alpha$ and $\beta$ can be determined by a hash function that can be stored in $\tilde{O}(1)$ space such that each entry can be constructed in $\tilde{O}(1)$ time. 
\item 
Given query partition $\calC=\{C_1,C_2, \ldots\}$, return
$
X=\left (Y- \sum_{\ell} \left (\sum_{i\in C_\ell} \alpha_i\right ) \left ( \sum_{i\in C_\ell} \beta_i\right )  \right )^2$.
\end{enumerate}

To analyze the algorithm we will need the following lemma due to  \citet{IndykM08} and \citet{BravermanCLMO10}.

\begin{lemma} For each 
$\{f_{ij}\}_{i,j\in [n]}$, $\expec{(\sum_{i,j} \alpha_i\beta_j f_{ij})^2}=\sum_{i,j}f_{ij}^2$ 
and $\var{(\sum_{i,j} \alpha_i\beta_j f_{ij})^2}\leq 9(\sum_{i,j}f_{ij}^2)\,^2 $.
\label{lemma:IndykM}
\end{lemma}

The following theorem will be proved by considering an algorithm that computes multiple independent copies of the above sketch and combines the estimates from each.

\begin{theorem}\label{firstds}
For unit weights,
there exists an $O(\epsilon^{-2} \log \delta^{-1} \log n)$-space algorithm for
the $\disagree$ query problem. 
Each positive edge is processed in
$\tilde{O}(\epsilon^{-2})$ time, while the query time is $\tilde{O}(\epsilon^{-2} n)$.
\end{theorem}
\begin{proof}
We first observe that, given $Y$, the time to compute $X$ 
is~$\tilde{O}(n)$. This follows because for a cluster $C_\ell \in \C$, on $n_\ell$~nodes, we can compute $\sum_{i\in C_\ell} \alpha_i$ and $\sum_{i\in C_\ell} \beta_i$ in $\tilde{O}(n_\ell)$ time. Hence the total query timeis $\tilde{O}(\sum_\ell n_\ell)=\tilde{O}(n)$ as claimed.  

We next argue that repeating the above scheme a small number of times
in parallel yields a good estimate of $\disagree(G,\calC)$.  To do this, note that
\[
X
=\left (\sum_{ij\in E^+} \alpha_i \beta_j - \sum_{\ell} \left (\sum_{i\in C_\ell} \alpha_i\right ) \left ( \sum_{i\in C_\ell} \beta_i\right )  \right )^2
=\left (\sum_{ij} \alpha_i \beta_j (M_{ij}^G-M_{ij}^\calC)  \right )^2
\ .
\]
We then apply  Lemma~\ref{lemma:IndykM} to $f_{ij}=M_{ij}^G-M_{ij}^\calC$ and deduce that 
\[\expec{X}=\disagree(G,\calC) ~~\mbox{ and }~~ \var{X}\leq 9 (\disagree(G,\calC))^2 \ .\]
Hence, running $O(\epsilon^{-2} \log \delta^{-1})$ parallel repetitions
of the scheme and averaging the results
appropriately 
yields a $(1\pm \epsilon)$-approximation for $\disagree(G,\calC)$  with probability at
least~$1-\delta$. Specifically, we partition the estimates into $O(\log \delta^{-1})$
groups, each of size $O(\epsilon^{-2})$.
We can ensure that with probability at least $2/3$, the mean of each group is within
a $1\pm \epsilon$ factor by an application of the Chebyshev bound; we then argue using the Chernoff bound that  the median of the resulting
group estimates is a $1\pm \epsilon$ approximation with probability at least $1-\delta$.
\end{proof}

\paragraph{Remark.} We note that 
by setting $\delta=1/n^n$ in the above theorem, it follows that we may estimate $\disagree(G,\calC)$ for all partitions $\calC$ using $\tilde{O}(\epsilon^{-2} n)$
space.
Hence,
given exponential time,
we can also $(1+\epsilon)$-approximate $\mindisagree(G)$
While this is near-optimal in terms of space,
in this paper we focus on polynomial-time algorithms.

\paragraph{Application to Cluster Repair.}
Consider the Cluster Repair problem~\citep{GrammGHN05},
in which, for some  constant~$t$,
we are promised $\mindisagree(G)\leq t$ and want to find the clustering $\argmin_{\C} \disagree(G,\C)$. 

We first argue that, given spanning forest $F$ of $(V,E^+)$ we can limit our attention to checking a polynomial number of possible clusterings. The spanning forest $F$ can be constructed using a $\tilde{O}(n)$-space  algorithm in the dynamic graph stream model~\citep*{AhnGM12}. Let $\calC_F$ be the clustering corresponding to the connected components of $E^+$.
Let $F_1, F_2,\ldots , F_p$ be the forests that can be generated by adding $t_1$ and then removing $t_2$  edges from $F$ where $t_1+t_2\leq t$. Let $\calC_{F_i}$ be the node-partition corresponding to the connected components of~$F_i$.

\begin{lemma}\label{lem:partitions}
The optimal partition of~$G$ is $\calC_{F_i}$ for some $1\leq i\leq p$. Furthermore, $p= O(n^{2t})$.
\end{lemma}
\begin{proof}
Let $E_*^+$ be the set of edges in the optimal clustering that are between nodes in the same cluster and suppose that let $E_*^+=(E^+ \cup A) \setminus D$, i.e., $A$ is the set of positive edges that need to be added and $D$ is the set of edges that need to be deleted to transform $E^+$ into a collection of node-disjoint clusters. Since $\mindisagree(G)\leq t$, we know $|A|+|D|\leq t$. It is possible to transform $F$ into a spanning forest $F'$ of $E^+ \cup A$ by adding at most $|A|$ edges. It is then possible to generate a spanning forest of $F''$ with the same connected components as  $E_*^+=(E^+ \cup A) \setminus D$ by deleting at most $|D|$ edges from $F'$. Hence, one of the forests $F_i$ considered has the same connected components at $E_*^+$.

To bound $p$, we proceed as follows. There are less than $n^{2t_1}$ different forests that can result from adding at most $t_1$ edges to $F$. For each, there are at most $n^{t_2}$ forests that can be generated by deleting at most $t_2$ edges from the, at most $n-1$, edges in~$F'$.  Hence, $p< \sum_{t_1,t_2: 0 \leq t_1+t_2\leq t} n^{2t_1+t_2} < t^2n^{2t}$.
\end{proof}

The procedure is then to take advantage of this bounded number of
partitions by computing each $\calC_{F_i}$ in turn, and estimating
$\disagree(G,\calC_{F_i})$.
We report the $\calC_{F_i}$ that minimizes the (estimated) repair cost. 
Consequently,  setting $\delta=1/(p \poly(n))$  in Theorem~\ref{firstds} yields
the following theorem.

\begin{theorem}\label{thm:smalldisagree}
For a unit-weight graph~$G$ with $\mindisagree(G)\leq t$ where $t=O(1)$,
there exists a
polynomial-time data-stream algorithm using $\tilde{O}(n + \epsilon^{-2} t )$
space that with high probability $1+\epsilon$ approximates $\mindisagree(G)$.
\end{theorem}

\subsection{Second Data Structure: Sparsification}
\label{sec:2ds}

The next data structure is based on graph sparsification and works for arbitrarily weighted graphs.
A sparsification of graph~$G$ is a weighted graph~$H$
such that the weight of every cut in~$H$ is within a $1+\epsilon$
factor of the weight of the corresponding cut in~$G$. A celebrated result of \citet{BenczurK96} shows that it is always possible to ensure the the number of edges in~$H$ is 
$\tilde{O}(n\epsilon^{-2})$. 
A subsequent result shows that this can be constructed in the
dynamic graph stream model.

\begin{theorem}[\citet{AhnGM12b,GoelKP12}]
\label{thm:sparse}
 There is a single-pass algorithm that returns a  sparsification using space~$\tilde{O}(n\epsilon^{-2})$ and time~$\tilde{O}(m)$.
\end{theorem}

The next lemma establishes that a graph sparsifier can be used to approximate $\agree$ and $\disagree$ of a clustering.

\begin{lemma}\label{lem:preservation}
  Let~$H^+$ and~$H^-$ be sparsifications
  of~$G^+=(V,E^+)$ and~$G^-=(V,E^-)$ such that all cuts are preserved within factor $(1\pm\epsilon/6)$, and let $H=H^+\cup H^-$.
For every clustering~$\calC$, 
\[\agree(G,\calC)=(1\pm \epsilon/2) \agree(H,\calC)\pm \epsilon w(E^+)/2
\]
and
\[\disagree(G,\calC)=(1\pm \epsilon/2)
\disagree(H,\calC)\pm \epsilon w(E^-)/2 \,. \]
Furthermore, $\maxagree(G)=(1\pm \epsilon)\maxagree(H)$.
\end{lemma}

\begin{proof}
The proofs for $\agree$ and $\disagree$ are symmetric,
so we restrict our attention to $\agree$. Let $\epsilon'=\epsilon/6$. The weight of edges in $E^-$ that are cut is within a $1+\epsilon'$ factor in the sparsifier. 
Consider an arbitrary cluster $C\in \cal C$, then letting~$w'(\cdot)$ represent the
weight in the sparsifier,
\begin{align*} w(uv\in E^+:u,v\in C) 
&=
w(uv\in E^+:u\in C,v\in V)
-w(uv\in E^+:u\in C,v\not \in C)\\
&=
\sum_{u\in C} w(uv \in E^+:v\in V)
-\sum_{u\in C} w(uv\in E^+:v\not \in C)\\
&=
(1\pm \epsilon')\sum_{u\in C} w'(uv \in E^+:v\in V) 
-(1\pm \epsilon') \sum_{u\in C} w'(uv\in E^+: v\not \in C)\\
&= 
w'(uv\in E^+:u,v\in C)
\pm 2 \epsilon' w'(uv\in E^+:u\in C,v\in V)\,,
\end{align*}
where the third line follows because, for each $u\in \calC$,
the weights of cuts $(\{u\},V\setminus \{u\})$ and $(C,V\setminus C)$ 
are approximately preserved.
Summing over all cluster $C \in \cal C$s, the total additive error is
\[2\epsilon' w'(E^+) \leq 2\epsilon' (1+\epsilon') w(E^+)\leq \epsilon
w(E^+)/2 \ ,\] (assuming $\epsilon \le 1$), as required.

The last part of theorem follows because $w(E^+)\leq \maxagree(G)$ by considering the trivial all-in-one-cluster partition. 
\end{proof}

\paragraph{Application to \maxagree\ with Bounded Weights.} 
In Section \ref{sec:lpsdp},
based on the sparsification construction,
we develop a $\poly(n)$-time streaming algorithm
that returns a~$0.766$-approximation for $\maxagree$
when~$G$ has arbitrary weights.
However, in the case of unit weights, a RAM-model  PTAS for $\maxagree$ is known \citep{GiotisG06,BansalBC04}. It would be unfortunate if, by approximating the unit-weight graph by a weighted sparsification, we lost the ability to return a $1\pm \epsilon$ approximation in polynomial time.

We resolve this by emulating an algorithm by \citet{GiotisG06}
for $\maxagree_k$ using a single pass over the stream\footnote{Note
  $\maxagree_k(G)\geq (1-\epsilon)\maxagree(G)$ for $k=O(1/\epsilon)$
  \citep{BansalBC04}.}.
Their algorithm is as follows:

\begin{enumerate}
\item Let $\{V^i\}_{i\in [m]}$ be an arbitrary node-partition,
where $m=\roundup{4/\epsilon}$ and $\rounddown{n/m}\leq |V^i|\leq \roundup{n/m}$. 
\item For each $j\in [m]$, let $S^j$ be a random sample of
$r=\poly(1/\epsilon,k,\log 1/\delta)$ nodes in $V\setminus V_j$.
\item For all possible $k$-partition of each of $S^1, \ldots, S^m$ : 
\begin{itemize}
\item For each $j$, let $\{S_i^j\}_{i\in k}$ be the partition of~$S^j$
\item Compute and record the cost of the clustering in which $v\in V^j$
is assigned to the~$i$th cluster, where
\[
i=\argmax_{i}\left ( \sum_{s\in S^j_i:\ sv\in E^+} w_{sv}+ \sum_{s\not \in
S^j_i:\ sv\in E^-} |w_{sv}| \right )\,.
\]
\end{itemize}
\item For all the clusterings generated, return the clustering $\calC$ that maximizes $\agree(G,\calC)$.
\end{enumerate}

\citet{GiotisG06} prove that the above algorithm achieves a $1+\epsilon$ approximation factor with high probability if all weights are $\{-1,+1\}$. We explain in Section \ref{sec:ggext} that their analysis actually extends to the case of bounded weights. The more important observation is that we can simulate this algorithm in conjunction with a graph sparsifier. Specifically, the sets $V_1, \ldots, V_m$ and $S_1, \ldots, S_m$ can be determined before the stream is observed. To emulate step 3, we just need to collect the $rnm$ edges incident to each $S_i$ during the stream. If we simultaneously construct a sparsifier during the stream we can evaluate all of the possible clusterings that arise. This leads to the following theorem.

\begin{theorem}
\label{thm:ptasboundedagree}
For bounded-weight inputs,
there exists a polynomial-time semi-streaming algorithm that,
with high probability,
$(1+\epsilon)$-approximates $\maxagree(G)$.
\end{theorem}

\subsection{Third Data Structure: Node-Based Sketch}
\label{sec:3ds}
In this section, we develop a data structure that supports queries to
$\disagree(G,\C)$ for arbitrarily {\em weighted} graphs when~$\C$ is
restricted to be a 2-partition.
For each node~$i$, define the vector, $a^i\in {\mathbb R}^{{n \choose 2}}$,
indexed over the ${n\choose 2}$ edges,
where the only non-zero entries are:

\[a^i_{ij}=
\begin{cases}
w_{ij}/2 & \mbox{if $ij\in E^-$} \\
w_{ij}/2 & \mbox{if $ij\in E^+$, $i<j$} \\
-w_{ij}/2 & \mbox{if $ij\in E^+$, $i>j$} \\
\end{cases} 
\]

\begin{lemma}\label{lem:2nodepres}
For a two-partition $\calC=\{C_1,C_2\}$,
$\disagree(G,\calC)=
\|\sum_{\ell\in C_1} a^\ell - \sum_{\ell\in C_2} a^\ell\|_1$.
\end{lemma}

\begin{proof}
The result follows immediately from consideration of the different possible for values for the $\{i,j\}$th coordinate of the vector $\sum_{\ell\in C_1} a^\ell - \sum_{\ell\in C_2} a^\ell$.
The sum can be expanded as
\[
\left | \left (\sum_{\ell \in C_1} a^\ell - \sum_{\ell\in C_2} a^\ell \right )_{ij}\right | =
\begin{cases}
\left |w_{ij}/2-w_{ij}/2\right | & \mbox{if $ij\in E^-$ and $i,j$ in different clusters} \\
\left |w_{ij}/2+w_{ij}/2\right | & \mbox{if $ij\in E^-$ and $i,j$ in the same cluster} \\
\left |w_{ij}/2+w_{ij}/2\right | & \mbox{if $ij\in E^+$ and $i,j$ in different clusters} \\
\left |w_{ij}/2-w_{ij}/2\right | & \mbox{if $ij\in E^+$ and $i,j$ in the same cluster} \\
\end{cases} \ .
 \]
Hence 
$\left | \left (\sum_{\ell \in C_1} a^\ell - \sum_{\ell\in C_2} a^\ell \right )_{ij}\right |=|w_{ij}|$
if and only if the edge is a disagreement.
\end{proof}

\noindent
We apply the $\ell_1$-sketching result of \citet*{KaneNW10} to compute a random linear sketch of each~$a^i$.

\begin{theorem}\label{thirdds}
For arbitrary weights, and for query partitions that contain two clusters, to solve the $\disagree$ query problem,
there exists an $O(\epsilon^{-2} n \log \delta^{-1} \log n)$-space algorithm.
The query time is $O(\epsilon^{-2} n \log \delta^{-1} \log n)$.
\end{theorem}
\noindent
Unfortunately, for  queries~$\calC$ where~$|\C|>2$,
$\Omega(n^2)$ space is necessary, as shown in Section~\ref{sec:lower}.

\paragraph{Application to $\mindisagree_2(G)$ with Bounded Weights.}
We apply the above node-based sketch in
conjunction with another algorithm by~\citet{GiotisG06},
this time for $\mindisagree_2$. Their algorithm is as follows:
\begin{enumerate}
\item Sample $r=\poly(1/\epsilon,k)\cdot \log n$ nodes $S$ and for every possible $k$-partition $\{S_i\}_{i\in [k]}$  of $S$: 
\begin{enumerate}
\item Consider the clustering where $v\in V\setminus S$ is assigned to the $i$th cluster where
\[
i=\argmax_{j}\left ( \sum_{s\in S_j: sv\in E^+} w_{sv}+ \sum_{s\not \in S_j: sv\in E^-} |w_{sv}| \right )
\]
\end{enumerate}
\item For all the clusterings generated, return the clustering $\calC$ that minimizes $\disagree(G,\calC)$.
\end{enumerate}
As with the max-agreement case, \citet{GiotisG06} prove that the above algorithm achieves a $1+\epsilon$ approximation factor with high probability if all weights are $\{-1,+1\}$. We explain in Section \ref{sec:ggext} that their analysis actually extends to the case of bounded weights. 
Again note we can easily emulate this algorithm for $k=2$ in the data stream model in conjunction with the third data structure. The sampling of~$S$ and its incident edges can be performed using one pass and $O(nr \log n)$ space. We then find the best of these possible partitions in post-processing using the above node-based sketches.

\begin{theorem}
\label{thm:boundeddisagree2}
For bounded-weight inputs, 
there exists a polynomial-time semi-streaming algorithm that
with high probability
$(1+\epsilon)$-approximates $\mindisagree_2(G)$.
\end{theorem}

\allowdisplaybreaks
\section{Convex Programming in Small Space: \mindisagree}
\label{sec:convexsmall}
\label{sec:lpsdp}

In this section, we present a
linear programming-based algorithm for \mindisagree.
At a high level, progress arises from new ideas and modifications
needed to implement convex programs in small space. 
While the time required to solve convex programs has always been an issue,
a relatively recent consideration is the restriction to small space~\citep{AhnG13BM}.
In this presentation, we
pursue
the Multiplicative Weight Update technique and its derivatives.
This method has a rich history across many different communities~\citep{AroraHK05},
and has been extended to semi-definite programs~\citep{AroraK07}. 
In this section, we focus on linear programs in the context of \mindisagree; we postpone the discussion of SDPs to Section~\ref{sec:sdp}.

In all multiplicative weight approaches, the optimization problem is first reduced to
a decision variant, involving a \emph{guess},~$\alpha$, of the objective value; 
we show later how to instantiate this guess.
The LP system is
\[\mbox{\mwmlp:} \left\{ \begin{array}{r l}
    & \bc^\mathrm{T}\by  \geq \alpha \\
    \mbox{s.t.}       & \bA\by \leq \bb, \quad \by \geq {\mathbf 0}\,,
  \end{array} \right.
\]
where $\bA \in {\mathbb R}^{N\times M}_+$,  $\bc,\by \in {\mathbb R}^{M}_+$,
and $\bb \in {\mathbb R}^N_+$.  
To solve the \mwmlp~approximately, the multiplicative-weight update
algorithm proceeds iteratively.
In each iteration, given the \emph{current} solution, $\by$, the procedure
maintains a set of multipliers
(one for each constraint) and computes a new \emph{candidate} solution~$\by'$ which
(approximately) satisfies the linear combination of the inequalities, as defined in Theorem~\ref{thm:mwmlp}.

\begin{theorem}[\cite{AroraHK05}] Suppose that, $\delta\leq \frac12$ and
in each iteration~$t$,
given a
vector
of non-negative  multipliers~$\bu(t)$,
a procedure (termed Oracle) provides a candidate~$\by'(t)$ satisfying three {\bf admissibility} conditions,
\begin{enumerate}[(i)]
\item
$\bc^T\by'(t) \geq \alpha$;
\item
$\bu(t)^T \bA \by'(t) - \bu(t)^T\bb \leq
\delta \sum_{i} \bu_i(t)$; and
\item
$-\rho \leq -\ell \leq \bA_i\by'(t) - \bb_i \leq \rho$, \quad for all $1\leq i \leq n$.
\end{enumerate}
We set $\bu(t)_i$ to $\bu(t+1)=(1+\delta( \bA_i\by'(t) - \bb_i )/\rho)\bu(t)$.
Assuming we start with $\bu(0)={\mathbf 1}$, 
after $T=O(\rho \ell \delta^{-2} \ln M)$ iterations
the average vector, $\by=\sum_t \by'(t)/ T$, satisfies $ \bA_i\by - \bb_i \leq 4\delta$, for all~$i$.
\label{thm:mwmlp}
\end{theorem}

The computation of the new candidate depends on the specific LP being solved.
The parameter~$\rho$ is called the \emph{width},
and controls the speed of convergence.
A small-width Oracle is typically a 
key component of an efficient solution, for example, to minimize running times, number of rounds, and so forth.
However, the width parameter is inherently tied to the specific formulation chosen.
Consider the standard LP relaxation for \mindisagree,
where variable~$x_{ij}$ indicates edge~$ij$ being cut.
\[ 
\begin{array}{l} 
\displaystyle \min \sum_{ij \in E^+} w_{ij}x_{ij} + \sum_{ij \in E^-} |w_{ij}| (1-x_{ij}) \\
x_{ij}+x_{j\ell} \geq x_{i\ell} \qquad \text{for all \ } i,j,\ell \\
x_{ij} \geq 0 \qquad \qquad \quad \   \text{for all \ } i,j
\end{array}
\]
The triangle constraints state that
if we cut one side of a triangle, we must also cut at least one of the other two sides.
The size of formulation is in~$\Theta(n^3)$, where $n$ is the size of the vertex set, irrespective of the number of
nonzero entries in $E^+ \cup E^-$.
Although we will rely on the sparsification of~$E^+$, that does not in any way
change the size of the above linear program.
To achieve~$\tilde{O}(n)$ space, we need new formulations, and new algorithms to solve them. 

The first hurdle is the storage requirement. We cannot store all the edges/variables which can be $\Omega(n^2)$. This is avoided by using a sparsifier and invoking (the last part of) Lemma~\ref{lem:preservation}.
Let~$H^+$ be the sparsification of~$E^+$ with $m'=|H^+|$.
 For edge $sq \in H^+$ let~$w^h_{sq}$  denote its weight after 
sparsification.
 For each pair $ij \in E^-$, let $P_{ij}(E')$ denote the set of all paths
involving edges only in the set~$E'$. Consider the following LP for \mindisagree, similar
to that of~\cite{Wirth04}, but in this sparsified setting:
\begin{align*}
 \hspace*{-0.5in} \displaystyle \mbox{min}&
\quad
\sum_{ij\in E^-} |w_{ij}| z_{ij}
                +\sum_{sq\in H^+} w^h_{sq}\, x_{sq} \\
 \hspace*{-0.5in}
    \forall p\in P_{ij}(H^+),ij\in E^-&
\quad
\displaystyle z_{ij}+\sum_{sq\in p} x_{sq}\geq 1 
 \lptag\label{lp:min-intuit} 
\\
 \hspace*{-0.5in}
\forall ij \in E^-, sq \in H^+ &
\quad
  z_{ij},x_{sq} \geq 0
\end{align*}
The intuition of an integral ($0$/$1$) solution is that~$z_{ij}=1$ for all
edges $ij \in E^-$ that are not cut, and $x_{sq}=1$ for all $sq \in H^+$
that are cut.
That is, the relevant variable is~$1$ whenever the edge disagrees with the input
\emph{advice}.
By Lemma~\ref{lem:preservation}, the objective value of~\ref{lp:min-intuit} is at most $(1+\epsilon)$ times the optimum value of \mindisagree.
However,~\ref{lp:min-intuit} now has exponential size,
and it is unclear how we can maintain the multipliers and update them in small space.
To overcome this major hurdle, we follow the approach below.

\subsection{A Dual Primal Approach}
\label{sec:lpdp}
Consider a primal minimization problem, for example, \mindisagree, in the canonical form:
\[
\mbox{Primal LP:} \left\{ \begin{array}{r l}
    \min  & \bb^\mathrm{T}\bx \\
    \mbox{s.t.}       & \bA^\mathrm{T}\bx \geq \bc, \quad \bx \geq {\mathbf 0}\,.
  \end{array} \right.
\]
The dual of the above problem for a guess,~$\alpha$ of the optimum solution
(to the Primal)
becomes 
\[\mbox{Dual LP:} \left\{ \begin{array}{r l}
    & \bc^\mathrm{T}\by  \geq \alpha \\
    \mbox{s.t.}       & \bA\by \leq \bb, \quad \by \geq {\mathbf 0}\,,
  \end{array} \right.
\]
which is the same as the decision version of \mwmlp\ as described earlier.
We apply Theorem~\ref{thm:mwmlp} to the Dual LP,
however we still want a \emph{solution} to the Primal LP.
Note that despite \emph{approximately} solving the
Dual LP, we do not have a Primal solution.
Even if we had some \emph{optimal} solution to the Dual LP, we might still
require a lot of space or time to find a Primal solution,
though we could at least rely on complementary slackness conditions.
Unfortunately,
similar general 
conditions do not exist for approximately optimum (or feasible) solutions. 
To circumvent this issue:

\begin{enumerate}[(a)]\parskip=0in
\item We apply the multiplicative-weight framework to
the Dual LP and try to find an approximately feasible solution~$\by$ such
that $\bc^T \by \geq (1-O(\delta))\alpha$ and $\bA\by \leq \bb,\by \ge 0$.

\item The Oracle is modified to provide a $\by$, subject to conditions (i)-- (iii) of Theorem~\ref{thm:mwmlp},
or an~$\bx$ that, for some~$f \geq 1$,
satisfies
$$\bb^T \bx \leq f \cdot \alpha, \quad
\bA^T\bx\geq \bc, \quad \bx\geq 0\,.$$
Intuitively, the Oracle is asked to either make progress towards finding a feasible dual solution or provide an
$f$-approximate primal solution in a single step. \label{oraclestep}
\item If the Oracle returns an~$\bx$ then we know that 
$\bc^T \by > (\bb^T \bx)/f$
is not satisfiable.
We can then consider 
smaller values of~$\alpha$, say $\alpha \leftarrow \alpha/(1+\delta)$. 
We eventually find a sufficiently small~$\alpha$ that the Dual LP is 
(approximately feasible) and we have a~$\bx$ satisfying
$$\bb^T \bx \leq f \cdot (1+\delta) \alpha, \quad \bA^T\bx\geq \bc, \quad
\bx\geq 0\,.$$
Note that computations for larger $\alpha$ continue to remain valid for smaller $\alpha$.
\end{enumerate}
This idea, of applying the multiplicative-weight update method to a
formulation with exponentially many variables (the Dual),
and modifying the Oracle to provide a solution to the Primal (that has exponentially many constraints) in a single step,
has also benefited solving \textsc{Maximum Matching} in small space~\citep{access}. However in \cite{access}, the constraint matrix was unchanging across iterations (objective function value did vary) -- here we will have the constraint matrix vary across iterations (along with value of the objective function). Clearly, such a result will not apply for arbitary constraint matrices and the correct choice of a formulation is key.
 
One key insight is that the dual, in this case (and as a parallel with matching) has exponentially many variables, but fewer constraints.
Such a constraint matrix is easier to satisfy approximately in a few iterations because there are many more \emph{degrees of freedom}.
This reduces the adaptive nature of the solution, and therefore we can make a lot of progress in satisfying many of the primal constraints in parallel. Other examples of this same phenomenon are the numerous dynamic connectivity/sparsification results in \cite{GMT}, where the algorithm repeatedly finds edges in cuts (dual of connectivity) to demonstrate 
connectivity. In that example, the $O(\log n)$ seemingly adaptive iterations collapse into a single iteration.

Parts of the three steps, that is, (a)--(c) outlined above,
have been used to speed up running times of SDP-based approximation algorithms~\citep{AroraK07}. 
In such cases, there was  no increase to the number of constraints nor consideration of non-standard formulations,
It is often thought, and as explicitly discussed by~\citet{AroraK07},
that primal-dual approximation algorithms use a  different set of
techniques from the primal-dual approach of multiplicative-weight update methods.
By switching the dual and the primal, in this paper, we align both sets of techniques and use them interchangeably. 

The remainder of Section~\ref{sec:convexsmall} is organized as follows.
We first provide a generic Oracle construction algorithm for \mwmlp, in Section~\ref{sec:generic}.
As a warm up example, we then apply this algorithm on the {\em multicut} problem in Section~\ref{sec:multicut} -- the
multicut problem is inherently related to \mindisagree\ for arbitrary weights~\citep{CharikarGW05,DemaineEFI06}.
We then show how to combine all the ideas together to solve \mindisagree\ in Section~\ref{sec:mindisagree}.

\subsection{From Rounding Algorithms to Oracles}
\label{sec:generic}
Recall the formulation \mwmlp, and Theorem~\ref{thm:mwmlp}.
Algorithm~\ref{alg:genericoracle} takes an~$f$-approximation for the Primal LP and produces an Oracle for \mwmlp.

\begin{algorithm}
 \begin{algorithmic}[1]
  \STATE Transform vector~$\bu(t)$ (a vector of weights for the constraints of \textbf{Dual LP})
   into a vector of scaled primal variables~$\bx$, thus:
$x_i = \alpha u(t)_i/\sum_i b_i u(t)_i$. 
  \STATE Perform a rounding algorithm for the Primal LP with~$\bx$ as the input fractional solution (as described in (b) previously). 
   Either there is a subset of violated constraints in the Primal LP
or (if no violated constraint exists) there is a solution with
   objective value at most~$f \cdot \alpha $, where~$f$ is the approximation factor for the rounding
algorithm. In case no violated constraint exists, return~$\bx$.
  \STATE Let $S=\{i_1,i_2,\ldots,i_k\}$ be (the indexation of) the set of
violated constraints in the Primal LP and
   let $\Delta=\sum_{i\in S} c_i$.
  \STATE Let $y_i=\alpha/\Delta$ for $i\in S$, and let $y_i=0$ otherwise.
   Return~$\by$. Note the two return types are different based on progress made in primal or dual directions.
 \end{algorithmic}
 \caption{From a rounding algorithm to an Oracle.\label{alg:genericoracle}}
\end{algorithm}
The following lemma shows how to satisfy the first two conditions of
Theorem~\ref{thm:mwmlp}; the width parameter has to be bounded separately for a particular problem.
\begin{lemma}\label{lem:admissible}
If $c_j>0$ for each Primal constraint, and $\sum_i u(t)_i >0$, then Algorithm~\ref{alg:genericoracle}
returns a candidate~$\by$ that satisfies conditions~(i) and~(ii) of Theorem~\ref{thm:mwmlp}.
\end{lemma}

\begin{proof}
 By construction, $\bc^\mathrm{T}\by=\alpha$, addressing condition~(i).
So we prove that $\bu(t)^T\bA\by - \bu(t)^T \bb \leq 0$.
Since~$\bu(t)$ is a scaled version of~$\bx$,
 \begin{align*}
  \frac{1}{\sum_i u(t)_i} \left( \bu(t)^T\bA\by - \bu(t)^T \bb \right)
   &  = \frac{1}{\sum_i x_i} \sum_i x_i (\bA_i\by - b_i)  = \frac{1}{\sum_i x_i} \left( \sum_i x_i \bA_i\by - \sum_i x_i b_i \right) \\
   &  = \frac{1}{\sum_i x_i} \left(\sum_j y_j (\bA^T_j \bx) - \sum_i x_i b_i \right)  \leq \frac{1}{\sum_i x_i} \left(\sum_j y_j c_j - \sum_i x_ib_i\right) 
     = 0
 \end{align*}
The inequality in the second line
follows from~$y_j$ only being positive if the corresponding Primal LP
constraint is violated.
Finally, by construction, $\sum_j y_j c_j=\alpha$ and $\sum_i b_i x_i=\alpha$;
since we also assumed that $\sum_i u(t)_i > 0$, the lemma follows.
\end{proof}

\renewcommand{\k}{\kappa}
\subsection{Warmup: Streaming \textsc{Multicut} Problem}
\label{sec:multicut}

The \textsc{Minimum Multicut} problem is defined as follows.
Given a weighted undirected graph and~$\k$ pairs of vertices~$(s_i,t_i)$,
for~$i=1,\ldots,\k$, the goal is to 
remove the lowest weight subset of edges such that every $i$, $s_i$ is disconnected from~$t_i$.

In the streaming context,
suppose that the weights of the edges are in the range~$[1,W]$
and the edges are ordered in an arbitrary order defining a dynamic data stream (with both insertions and deletions). 
We present a $O(\log \k)$-approximation algorithm for
the multicut problem that uses $\tilde{O}(n\epsilon^{-2} \log W + \k)$ space
and $\tilde{O}(n^2\epsilon^{-7}\log^2 W)$ time excluding the time to construct a sparsifier. The $\tilde{O}(n^2)$ term dominates the time required for sparsifier construction, for more details regarding streaming sparsifiers, see \cite{KapralovLMMS14,GMT}.
The algorithm comprises the following, the parameter $\delta$ will eventually be set to $O(\epsilon)$.
\begin{enumerate}[{\sc MC}1]
\item
Sparsify the graph defined by the dynamic data stream, preserving all cuts, and thus the optimum 
  multicut, within $1\pm\delta$ factor.
Let~$E'$ be the edges in the sparsification and~$|E'|=m'$, where
$m'=O(n\delta^{-2}\log W)$, from the results of~\citet{AhnGM12b}.
Let~$(w_{jq})$ refer to weights after the sparsification.

\item
Given an edge set $E'' \subseteq E'$,
let~$P'(i,E'')$ be the set of all $s_i$--$t_i$ paths in the edge set~$E''$.
The LP that captures \textsc{Multicut} is best viewed as relaxation of a $0/1$ assignment.
Variable~$x_{jq}$ is an indicator of whether edge~$(j,q)$ is in the multicut.
If we interpret~$x_{jq}$ as assignment of lengths, then for all $i \in [\k]$, all $p \in P'(i,E')$ have length at least~$1$.
The relaxation is therefore:

\begin{align*}
 \begin{array}{r l l}
  \multicolumn{2}{l}{\alpha^* = \min \sum_{(j,q)\in E'} w_{jq}x_{jq}} \\
  \mbox{s.t.} 
    & \sum_{(j,q)\in p} x_{jq}\geq 1 
    & \text{for all } i \in [\k], p \in P'(i,E') \\
& x_{jq} \geq 0 & \forall (j,q) \in E'
 \end{array}
 \lptag\label{lp:multicut} 
\end{align*}

\item Compute an initial upper bound $\alpha_0 \in [(1+4\delta)\alpha^*,(1+4\delta)n^2\alpha^*]$. (Lemma~\ref{initlemma})
\item
Following the dual-primal approach above, as $\alpha$ decreases (note the initial $\alpha_0$ being high, we cannot hope to even approximately satisfy the dual), we consider the (slightly modified) dual

\begin{align*}
 \begin{array}{r l l}
   & \sum_{p} y_p  \geq \alpha \\
   & \frac{1}{w_{jq}} \sum_{p: (j,q)\in p} y_p \leq 1
   & \mbox{for all }(j,q)\in E' \\
   & y_p \geq 0 
   & \text{for all } i \in [\k], p \in P'(i,E')
 \end{array}
 \lptag\label{lp:multicut-d}
\end{align*}
More specifically, we consider the following variation: given~$\alpha$,
let $E'(\alpha)$ be the set of edges of weight at least
$\delta\alpha/m'$, and we seek:

\begin{align*}
\hspace{-0.5cm}
 \begin{array}{r l l}
   & \sum_{p} y_p \geq \alpha \\
   & \frac{1}{w_{jq}} \sum_{p: (j,q)\in p} y_p \leq 1 
   & \mbox{for all }(j,q)\in E'(\alpha) \\
   & y_p \geq 0 
   & \forall p \in P'(i,E'(\alpha)), i 
 \end{array}
 \lptag\label{lp:multicut-d2}
\end{align*}

\item We run the Oracle is provided in Algorithm~\ref{alg:multicut}. 

\item If we receive a $\bx$ we set $\alpha \leftarrow \alpha/(1+\delta)$ as in (c) in Section~\ref{sec:lpdp}. This step occurs at least once (Lemma~\ref{qualitylemma}).
Note that reducing $\alpha$ corresponds to {\em adding constraints as well as variables} to \ref{lp:multicut-d2} due to new edges in $E'(\alpha/(1+\delta)) - E'(\alpha)$. We set $u_{i'}(t+1)=(1-\delta/\rho)^{t}$ for each new constraint $i'$ added, assuming that we have run the Oracle in step (MC5) a total of $t$ times thus far. Lemma~\ref{lem:multicutwidth} shows that this transformation provides a $\bu$ and a collection $\by(t)$ as if the multiplicative weight algorithm for \ref{lp:multicut-d2} was run for the current value of $\alpha=\alpha_1$.

\item If we have completed the number of iterations required by Theorem~\ref{thm:mwmlp} we average the $\by$ returned then we have an approximately feasible solution for \ref{lp:multicut-d2}. This corresponds to a proof of (near) optimality. 
We return the $\bx$ returned corresponding to the previous value of $\alpha$ (which was $\alpha(1+\delta))$ as the solution. This is a $f(1+O(\delta))$ approximation (Lemma~\ref{qualitylemma}). If we have not completed the number of iterations, we return to (MC5).
\end{enumerate}

\begin{lemma}\label{initlemma}
Consider introducing the edges of~$E'$ from the largest weight to smallest.
Let $w$ be the weight of the first edge whose introduction connects some pair $(s_i,t_i)$. Set $\alpha_0=(1+4\delta)n^2w$. Then
$\alpha_0 \in [(1+4\delta)\alpha^*,(1+4\delta)n^2\alpha*]$.
\end{lemma}
\begin{proof}
Note $w$ is a lower bound on $\alpha^*$; moreover, if we delete the edge with weight $w$ and all subsequent edges in the ordering we have feasible multicut solution. Therefore $\alpha^* \leq n^2 w$. The lemma follows.

Naively, this edge-addition process runs in~$\tilde{O}(m' \k)$ time, since the connectivity needs to be checked for every pair.
However, we can introduce the edges in groups, corresponding to weights in $(2^{z-1},2^z]$, as~$z$ decreases; we check
connectivity after introducing each group.
This algorithm runs in time $\tilde{O}(m'+\k \log W)$ and approximates~$w$, i.e., overestimates~$w$ by a factor of at
most~$2$,
since we have a geometric sequence of group weights.
The initial value of~$\alpha$ can thus be set to~$(1+4\delta)2^z n^2$.
\end{proof}

\begin{lemma}
\label{qualitylemma}
$\alpha$ is decreased, as in (MC6), at least once.
The solution returned in (MC7) is a $f(1+O(\delta))$ approximation to $\alpha^*$. 
\end{lemma}
\begin{proof}
Using Theorem~\ref{thm:mwmlp} once we are in (MC7) multiplying the average of the $y_p$ by $1/(1+4\delta)$ gives a feasible solution for \ref{lp:multicut-d2} for the edge set $E(\alpha)$. Moreover, 
for all paths $p$, containing any edge in $E' -E'(\alpha)$, we have $y_p=0$. Therefore this new solution is a feasible solution of \ref{lp:multicut-d}.  Therefore $\alpha/(1+4\delta) \leq \alpha^*$ once we reach the required number of iterations in (MC7).
This proves that we must decrease $\alpha$ at least once, because $\alpha_0$ is larger than $(1+4\delta)\alpha^*$ (Lemma~\ref{initlemma}).

 The solution $\bx$ corresponds to $f \alpha (1+\delta)$. Since $\alpha$ is bounded above by $\alpha^*(1+4\delta)$, the second part of the lemma follows as well.
\end{proof}

\begin{algorithm}
 \begin{algorithmic}[1]
  \STATE Given weights~$u^t_{jq}$, for $(j,q) \in E'(\alpha)$, define
$x_{jq} = \alpha u^t_{jq}/\sum_{(j,q) \in E'(\alpha)} w_{jq} u^t_{jq}$. 
  \STATE Define the shortest path metric $d^x(\cdot,\cdot)$ with the~$x_{jq}$ representing edge lengths.
Define $B(\zeta,r)=\{\zeta'\mid d^x(\zeta,\zeta') \leq r\}$, which corresponds to a family of balls/regions centered
at~$\zeta$, each of radius~$r$.
Let  $\cut(B(\zeta,r))$ be the total weight of edges in~$E'(\alpha)$ that are cut by $B(\zeta,r)$, i.e., 
 \[\cut(B(\zeta,r))=\sum_{(\zeta',\zeta'') \in E'(\alpha); d^x(\zeta,\zeta') \leq r < d^x(\zeta,\zeta'')} w_{\zeta' \zeta''} \,.\]
\STATE Find a collection of regions $B(\zeta_1,r_1),\ldots,B(\zeta_g,r_g),\ldots$ such that every $r_g \leq \frac13$ and
each~$s_i$ belongs to some region, and $\sum_g \cut(B(u_g,r_g)) \leq 3\alpha \ln (\k+1)$. Lemma~\ref{lem:ball} shows us
how to achieve this.
  \IF{for some~$i$ both~$s_i$ and~$t_i$ belong to the same region}
   \STATE Find the corresponding path~$p$, which is of length at most~$2/3$, which violates the constraint.
Return~$y_p=\alpha$.  Implicitly return~$y_{p'}=0$ for all other paths that involve the $s_i$--$t_i$ pair.
  \ELSE
   \STATE Return the union of the cuts defined by the balls (this corresponds to $\bx$). The edges in $E'(\alpha)$ contribute at most $3 \alpha \ln
(\k+1)$.
The edges in~$E' - E'(\alpha)$ contribute at most~$\delta \alpha$. The total is $(3\ln (\k+1)+\delta)\alpha$. Note that the return types are different as outlined in the dual-primal framework in (a)--(c) earlier.
  \ENDIF
 \end{algorithmic}
 \caption{\label{alg:multicut}Oracle for~\ref{lp:multicut-d2}}
\end{algorithm}

\begin{corollary}
We decrease $\alpha$ at most $O(\delta^{-1} \log n)$ times in step MC6.
\end{corollary}
\begin{proof}
If we decrease $\alpha$ then at some point line (7) of Algorithm~\ref{alg:multicut} provides a solution $\ll \alpha^*$, which is infeasible. Note that the solution would have value $f \alpha$. But this has to be at least $\alpha^*$. Thus $\alpha$ cannot decrease arbitrarily. Combined with the upper bound in Lemma~\ref{initlemma}, the result follows.
\end{proof}
\begin{lemma}
\label{lem:multicutwidth}
Algorithm~\ref{alg:multicut} returns an admissible (defined in Theorem~\ref{thm:mwmlp}) $\by$
 for \ref{lp:multicut-d2} with (the width)  $\rho=m'/\delta$ and~$\ell=1$.
 Moreover the set of assignments of~$y_p$ (over the different iterations) that were admissible
for~$\alpha = \alpha_2$ remains admissible if~$\alpha$ is lowered to~$\alpha_1 < \alpha_2$ and $\bu$ updated as dectribed in (MC6).
\end{lemma}

\begin{proof}
 Using Lemma~\ref{lem:admissible}, 
 Algorithm~\ref{alg:multicut} returns a $\by$ which satisfies conditions
 (i) and (ii) of Theorem~\ref{thm:mwmlp}.
 By construction, in Algorithm~\ref{alg:multicut} 
 $y_p=\alpha$ and only one~$y_p$ has a non-zero value. Since we removed all
 the edges of weight less than $\delta\alpha/m'$,
 the width parameter is bounded by ${\alpha m'}/({\delta\alpha})=m'/\delta$. Observe that $\ell=1$.
 
 If $\alpha_1 < \alpha_2$, then $E'(\alpha_1) \supseteq E'(\alpha_2)$,
and therefore $P(i,E'(\alpha_1)) \supseteq P(i,E'(\alpha_2))$. Therefore, for the formulation \ref{lp:multicut-d2}, we are adding new variables corresponding to new variables (paths) as well as new constraints corresponding to the newly added edges.
We can interpret the $\by$ for $\alpha_2$ to have $0$ values for the new variables. This would immediately satisfy (i). This would satisfy (iii) for the old constraints as well. Condition (iii) is satisfied for the newly introduced constraints because the old paths $p$ with $y_p>0$ for $\alpha_2$ did not contain an edge in $E'(\alpha_1)$. Thus $\bA_i \by(t) =0$ for the new constraints and $\bb = {\mathbf 1}$ and $-\rho \leq -1 \leq \rho$.

For (ii), $\bu(t)^T\bA \by(t) - \bu(t)^T\bb \leq \delta \sum_i \bu(t)$, the first term in the left hand side remains unchanged. The left hand side decreases for every new constraint, and the right hand side increases for every new constraint.
\end{proof}

\noindent The next lemma arises from a result of \citet{GargVY93};
in this context, $Z=\alpha$.

\begin{lemma}\cite{GargVY93}.\label{lem:ball}
Let $Z=\sum_{(u,v)} x_{uv} w_{uv}$. For $r\geq 0$, let $B(u,r)=\{v\mid d^x(u,v) \leq r \}$ where $d^x$ is the shortest path distance based on the values $x_{uv}$. Let $\vol(B(u,r))$ be
{\small 
\[ 
\frac{Z}{\kappa} + \sum_{\stackrel{(v,v')}{v,v' \in B(u,r)}} x_{vv'}w_{vv'} + \sum_{\stackrel{(v,v')}{v \in B(u,r),v' \not \in B(u,r)} }(r-d^x(u,v))w_{vv'}\]
}
Suppose that for a node~$\zeta$, the radius~$r$ of the ball around~$\zeta$ is increased until $\cut(B(\zeta,r)) \leq C
\cdot \vol(B(\zeta,r))$. If $C=3\ln (\kappa+1)$, the ball stops growing before the radius becomes~$1/3$.
We start this process for $\zeta_1=s_1$.
Repeatedly, if some~$s_j$ is not in a ball, then we remove $B(\zeta_i,r_i)$ (all edges inside and those being cut)
and continue the process with $\zeta_{i+1}=s_j$, on the remainder of the graph. The collection of
$B(\zeta_1,r_1),\ldots,B(\zeta_g,r_g), \ldots$ satisfy the condition that~$r_g \leq 1/3$ for all~$g$ and $\sum_g \cut(B(\zeta_g,r_g)) \leq CZ$.
\end{lemma}

The proof follows from the fact that $\cut(B(\zeta,r))$ is the derivative of $\vol(B(\zeta,r))$ as $r$ increases and the
volume cannot increase by more than a factor of $\kappa+1$,
because it is at least $Z/k$ and cannot exceed $Z/k+Z$.
For nonnegative~$x_{jq}$ the above algorithm runs in time~$\tilde{O}(m')$ using standard shortest-path algorithms.

Using Theorem~\ref{thm:mwmlp}, the total number of iterations needed in MC7, for a particular~$\alpha$ is $O(\rho\delta^{-2}\log N)=O(m' \delta^{-3} \log n)$, since the number of constraints $N=O(n^2)$ and $\rho \leq
m'/\delta$. This dominates the $O(\frac1\delta \cdot \log n)$ times we
decrease~$\alpha$.

Observe that the algorithm repeatedly constructs a set of balls with non-negative weights; which can be performed in $O(m'\log n)$ time. In each of these balls with $\tilde{m}$ edges, we can find the shortest path in $O(\tilde{m}\log n)$ time (to find the violated pair $s_i$--$t_i$). Summed over the balls, each iteration can be performed in $O(m'\log n)$ time.
Coupled with the approximation introduced by a sparsifier, setting $\delta=O(\epsilon)$ we get:

\begin{theorem}\label{thm:multicut}
 There exists a single-pass $O(\log \k)$-approximation algorithm
 for the multicut problem in the dynamic semi-streaming model that runs
 in $\tilde{O}(n^2\epsilon^{-7} \log^2 W)$ time and
 $\tilde{O}(n\epsilon^{-2} \log W+ \k)$ space.
\end{theorem}

\subsection{\mindisagree\  with Arbitrary Weights}
\label{sec:mindisagree}
In this section, we prove the following theorem:

\begin{theorem}\label{thm:min} 
There is a $3(1+\epsilon)\log |E^-|$-approximation
algorithm for $\mindisagree$  that requires 
 $\tilde{O}((n\epsilon^{-2} + |E^-|)^2\epsilon^{-3})$ time, $\tilde{O}(n\epsilon^{-2} + |E^-|)$ space, and a single pass.
\end{theorem}

Consider the dual of~\ref{lp:min-intuit}, where $\bP=\cup_{ij \in E^-} P_{ij}(H^+)$.

\begin{align*}
 & \max \sum_{p} y_p \\
 & \frac{1}{|w_{ij}|}\sum_{p \in P_{ij}(H^+) } y_p \leq 1 
    & \forall ij\in E^- \\
 &  \frac{1}{w^h_{sq}} \sum_{p \in \bP: sq \in p} y_p \leq 1 
    & \forall sq \in H^+ \\
& y_p \geq 0 & \forall p \in \bP
\lptag \label{mindual}
\end{align*}

We apply Theorem~\ref{thm:mwmlp} (the multiplicative-weight update framework)
to
the dual of \ref{lp:min-intuit}, but omit the constraints
in the dual corresponding to small-weight edges, exactly along the lines of MC1--MC7.
For each $\alpha \geq0$, let $H^+(\alpha),E^-(\alpha)$ be the set of edges in
$H^+,E^-$, respectively, with weight at least $\delta \alpha/(m'+|E^-|)$.
Consider:
\begin{align*}
 & \sum_{p} y_p \geq \alpha \\
 & \frac{1}{|w_{ij}|}\sum_{p \in P_{ij}(H^+(\alpha)) } y_p \leq 1 
    & \forall ij\in E^-(\alpha) \\
 &  \frac{1}{w^h_{sq}} \sum_{p \in \bP:sq \in p} y_p \leq 1 
    & \forall sq\in H^+(\alpha) \\
& y_p\, \geq 0 & \forall p \in \bP(\alpha)
\lptag\label{mindual2}
\end{align*}
where
$\bP(\alpha) = \bigcup_{ij \in {E}^-(\alpha)} P_{ij}(H^+(\alpha))$.

We attempt to find an approximate feasible solution to \ref{mindual2}
for a large value of~$\alpha$.
If the Oracle fails to make progress then it provides a solution to \ref{lp:min-intuit} of value $f \cdot \alpha$.
In that case we set $\alpha \leftarrow \alpha/(1+\delta)$ and
try the Oracle again.
Note that if we lower $\alpha$ then the Oracle invocations for larger values of $\alpha$ continue to remain valid; if $\alpha_1 \leq \alpha_2$, then
$P_{ij}(H^+(\alpha_1)) \supseteq P_{ij}(H^+(\alpha_2))$ exactly along the lines of Lemma~\ref{lem:multicutwidth}.

Eventually we lower~$\alpha$ sufficiently that we have a feasible solution
to \ref{mindual2}, and we can claim Theorem~\ref{thm:min} exactly along the lines of Theorem~\ref{thm:multicut}.
The Oracle is provided in Algorithm~\ref{alg:mindis1} and relies on the following lemma:

\begin{lemma}\label{lem:ball2}
Let $\kappa=|E^-|$, $Z=\sum_{uv \in H^+(\alpha)} x_{uv} w^h_{uv}$. Using the definition of $d^x()$ and $B()$ as in Lemma~\ref{lem:ball}, let 
\begin{align*}
\textup{vol}(B(u,r)) = & 
\frac{Z}{\kappa} + \sum_{\substack{vv' \in H^+(\alpha)\\ v,v' \in B(u,r)}}  x_{vv'}w^h_{vv'} 
 \quad  + \sum_{\substack{vv' \in H^+(\alpha) \\ d^x(u,v)\leq r<d^x(u,v')}} (r-d^x(u,v))w^h_{vv'} \ .
\end{align*}
Suppose that, for a node~$\zeta$, the radius~$r$ of its ball is increased
until $\textup{cut}(B(\zeta,r)) \leq C \textup{vol}(B(\zeta,r))$.
If $C=3\ln (\kappa+1)$, the ball stops growing before the radius
becomes~$1/3$.
We start this process setting $\zeta_1$ to be an arbitrary endpoint of an edge
in~$E^-$, and let the stopping radius be~$r_1$.
We remove $B(\zeta_1,r_1)$
and continue the process on the remainder of the graph.
The collection of $B(\zeta_1,r_1),B(\zeta_2,r_2), \ldots$ satisfy the condition that each radius is at most $1/3$  and $\sum_g \text{cut}(B(\zeta_g,r_g)) \leq CZ$.
\end{lemma}

\begin{algorithm}[t]
 \begin{algorithmic}[1]
  \STATE Given multipliers $u^t_{sq}$ for $sq \in H^+(\alpha)$ and $v^t_{ij}$ for $ij \in  E^-(\alpha)$, define $Q_u = \sum_{sq \in H^+(\alpha)} w^h_{sq} u^t_{sq}$ and $Q_v= \sum_{ij \in E^-(\alpha)} |w_{ij}| v^t_{ij}$. 
\STATE Let $x_{sq}= \alpha u^t_{sq}/(Q_u+Q_v)$, $z_{ij}= \alpha v^t_{ij}/(Q_u+Q_v)$.
  \STATE
Treating the~$x_{sq}$ as edge lengths,
let $d^x(\cdot,\cdot)$ be the shortest path metric.
Define $B(\zeta,r)=\{ \zeta' \mid d^x(\zeta,\zeta') \leq r\}$
and the weight of the edges of cut by the ball:
\[\text{cut}(B(\zeta,r))=\sum_{\substack{\zeta'\zeta'' \in H^+(\alpha) \\ d^x(\zeta,\zeta') \leq r < d^x(\zeta,\zeta'')}} w^h_{\zeta'\zeta''}\]
\STATE Find a collection of balls $B(\zeta_1,r_1),B(\zeta_2,r_2),\ldots $ such that 
(i) each radius at most $1/3$, (ii) every endpoint of an edge in $E^-(\alpha)$ belongs to
some ball, and (iii) $\sum_g \text{cut}(B(\zeta_g,r_g)) \leq 3\alpha {Q_u}/(Q_u+Q_v)
\cdot \ln (|E^-|+1)$. The existence of such balls follows from Lemma~\ref{lem:ball2}.
  \IF{there exists $ij \in E^-(\alpha)$ with $i,j$ in the same ball and $z_{ij}<1/3$.}
   \STATE Find the corresponding path~$p$ between $i$ and $j$. Since the length of this path is at most $2/3$ and $z_{ij}<1/3$, the corresponding constraint is violated. 
   Return $y_p=\alpha$ and $y_{p'}=0$ for all other paths for edges in $E^-$.
  \ELSE
   \STATE Return the union of cuts defined by the balls and all edges in $H^+ - H^+(\alpha)$.
  \ENDIF
 \end{algorithmic}
 \caption{\label{alg:mindis1}Oracle for \ref{mindual2}}
\end{algorithm}

The above lemma is essentially the same as Lemma~\ref{lem:ball}, applied to the terminal pairs defined by the endpoints of each edge in $E^-$.
Again, for nonnegative~$x_{sq}$, standard shortest-path algorithms lead to a running
time of $\tilde{O}(m')$.
We bound the width of the above oracle as follows :
\begin{lemma}
$\rho=(m'+|E^-|)/\delta$, $\ell=1$ for Algorithm~\ref{alg:mindis1}.
\end{lemma}

The total weight of positive edges cut by the solution returned in line~8 of
Algorithm~\ref{alg:mindis1} is at most ${3 \alpha Q_u}/(Q_u+Q_v) \cdot  \ln (|E^-|+1)$.
Each negative edge that is not cut corresponds to setting
$z_{ij}=1$ but $z_{ij} \geq 1/3 $; hence the cost of these edges is $\frac{3\alpha Q_v}{Q_u+Q_v}$.
Finally, the cost of the edges in neither~$E^-(\alpha)$ nor~$H^+(\alpha)$ is at most~$2\delta \alpha$. The overall solution has cost $(3 \ln (|E^-|+1) + 2\delta)\alpha$.

Finally, we show how to initialize $\alpha$ along the lines of Lemma~\ref{initlemma}.
Divide the edges of~$H^+$ according to weight, in intervals $(2^{z-1},2^z]$,
as we decrease~$z$.
For each group~$z$,
we find the largest weight edge $ij \in E^-$, call this weight $g(z)$, such that~$i$ and~$j$ are
connected by $H^+$-edges of group~$z$ or higher.
Observe that~$g(z)$ is an increasing function of
$z$.
Let the smallest~$z$ such that $g(z) \geq 2^z$ be $z_0$. Then it follows that
the optimum solution is at least~$2^{z_0-1}$. Again, $2^{z_0} n^2$ serves
as an initial value of~$\alpha$,
which is an~$O(n^2)$  approximation to the optimum solution.

\allowdisplaybreaks
\section{Convex Programming in Small Space: \maxagree}
\label{sec:sdp}
In this section we discuss an SDP-based algorithm for \maxagree.
We will build upon our intuition in Section~\ref{sec:convexsmall} 
where we developed a linear program based algorithm
for \mindisagree. However several steps, such as switching of primals
and duals, will not be necessary because we will use a modified
version of the multiplicative weight update algorithm for SDPs as
described by \cite{Steurer10}.
As will become clear, the switch of primals and duals is already achieved in the internal working of \cite{Steurer10}.
Consider:

\begin{definition}
\label{basicdef}
  For matrices $\bX,\bZ$, let $\bX \circ \bZ$ denote 
  $\sum_{i,j} \bX_{ij}\bZ_{ij}$, let 
  $\bX\succeq {\mathbf 0}$ denote that $\bX$ is positive semidefinite,
and let $\bX \succeq \bZ$ denote $\bX-\bZ \succeq {\mathbf 0}$.
\end{definition}

A semidefinite decision problem in canonical form is:
\[\mbox{MWM SDP:} \left\{ \begin{array}{r l}
    & \bC \circ \bX  \geq \alpha \\
    \mbox{s.t}       & \bF_j \circ \bX \leq g_j, \ \forall 1 \leq j \leq q, \quad \bX \succeq {\mathbf 0} 
  \end{array} \right.
\]
where $\bC,\bX \in {\mathbb R}^{n\times n}$ and $\bg \in {\mathbb R}^q_+$. Denote the set of the feasible solutions by $\mathcal{X}$.
 Typically we are interested in the Cholesky decomposition of $\bX$,
a set of~$n$ vectors~$\{\bx_i\}$
such that $\bX_{ij}=\bx^T_i  \bx_j$.
Consider the following theorem:

\begin{theorem}[\cite{Steurer10}] Let~$\bD$ be a fixed diagonal matrix with positive entries and assume ${\mathcal X}$ be nonempty.
Suppose there is an Oracle that for each positive semidefinite~$\bX$
either
(a) tests and declares~$\bX$ to be approximately feasible --- for all
$1\leq i \leq q$, we have $\bF_i \circ \bX \leq g_i +\delta$, or (b)
provides a real symmetric matrix~$\bA$ and a scalar~$b$ satisfying
(i) $\bA \circ \bX \leq b - \delta$ and for all $\bX' \in {\mathcal X}$,
$\bA \circ \bX' \geq b$ and (ii)
$ \rho \bD \succeq \bA - b\bD \succeq - \rho \bD$,
then a multiplicative-weight-style algorithm 
produces an approximately feasible $\bX$,
in fact its Cholesky decomposition,
in $T=O(\rho^2 \delta^{-2} \ln n)$ iterations.
\label{thm:mwmsdp}
\end{theorem}

The above theorem does not explicitly discuss maintaining a set of multipliers. But interestingly, the
algorithm in \cite{Steurer10} that proves Theorem~\ref{thm:mwmsdp}
can be viewed as a dual-primal algorithm. This algorithm collects separating hyperplanes
to solve the dual of the SDP: on failure to provide such a hyperplane, the algorithm provides a primal feasible~$\bX$.
The candidate~$\bX$ generated by the algorithm is an exponential
of the (suitably scaled) averages of the hyperplanes $(A,b)$: this
would be the case if we were applying the multiplicative-weight update
paradigm to the dual of the SDP in canonical form! Therefore, along with maximum matching \cite{access} and \mindisagree\ (Section~\ref{sec:convexsmall}) we have yet another example where switching the primal and the dual formulations helps. However in all of these cases, we need to prove that that we can produce a feasible primal solution in a space efficient manner, when the Oracle (for the dual) cannot produce a candidate.

We now prove the following theorem:

\begin{theorem}\label{thm:max}
 There is a $0.7666(1-\epsilon)$-approximation algorithm
 for $\maxagree(G)$ that uses $\tilde{O}(n\epsilon^{-2} )$ space,
$\tilde{O}(m+{n}\epsilon^{-10})$ time and a single pass.
\end{theorem}

We use Lemma~\ref{lem:preservation} and edge set $H=H^+ \cup H^-$. Let
$w^h_{ij}$ correspond to the weight of an edge $ij \in H$.
Our SDP for \maxagree\ is:
  \begin{align*}
& \sum_{ij\in H^+}w^h_{ij}\bX_{ij} 
   + \sum_{ij\in H^-} \frac{ |w^h_{ij}| (\bX_{ii} + \bX_{jj} - 2\bX_{ij})}{2} \geq \alpha \\
&   \begin{array}{r ll}
   \bX_{ii}  & \leq 1  &\forall  i\in V \\
   - \bX_{ii}  &\leq - 1  &\forall  i\in V \\
   - \bX_{ij} & \leq 0  & \forall  i,j\in V \\
   \bX & \succeq {\mathbf 0} &
   \end{array}
   \tag{SDP}\label{sdp:maxag}
  \end{align*}
If two vertices,~$i$ and~$j$, are in the same cluster, their corresponding vectors~$\bx_i$
and~$\bx_j$ will coincide, so $\bX_{ij}= 1$; on the other hand, if they
are in different clusters, their vectors
should be orthogonal, so $\bX_{ij}=0$.
Observe that under the restriction $\bX_{ii}=\bX_{jj}=1$, the contribution of an $ij \in H^-$ is $\bX_{ii} + \bX_{jj} - 2\bX_{ij}=(1-\bX_{ij})$,
as intended.
However, this formulation helps prove that the width is small.

\begin{definition}
\label{defD}
Define $d_i = \sum_{j:ij \in H} |w^h_{ij}|$ and $\sum_i d_i = 2W$. Let $\bD$ be the diagonal matrix with $\bD_{ii}=d_i/2W$.
\end{definition}

A random partition of the graph provides a trivial
$1/2$-approximation for maximizing agreements.
Letting~$W$ be the total weight of edges in~$H$, the sparsified graph, we  perform binary search for $\alpha \in [W/2,W]$, and stop when the interval is of size $\delta W$.
This increases the running time by a $O(\log \delta^{-1})$ factor.

The diagonal matrix $\bD$ specified in Definition~\ref{defD} sets up the
update algorithm of \cite{Steurer10}.
The choice of $\bD$ will be critical to our algorithm: typically, this $\bD$ determines the ``path'' taken by the SDP solver, since $\bD$ alters the projection to density matrices.
Summarizing, Theorem~\ref{thm:max} follows from the Oracle provided
in Algorithm~\ref{alg:maxoracle}. 
The final solution  
only guarantees $\bx_i \cdot \bx_j \geq - \delta$. Even though the standard rounding algorithm assumes $\bX_{ij} \geq 0$, the fractional solution with $\bX_{ij} \geq - \delta$ can be rounded efficiently.
Ensuring $ \bx_i \cdot \bx_j \geq 0$ appears to be difficult (or to require a substantially different oracle).

\begin{algorithm}[ht]
 \begin{algorithmic}[1]
  \STATE
For the separating hyperplane,
   we only describe non-zero entries in~$\bA$. Recall that we have a candidate $\bX$ where $\bX_{ij} = \bx_i \cdot \bx_j$.
  \STATE Let $S_\cgeq=\{i : \|\bx_i\|^2\geq 1+\delta\}$, 
         $\Delta_\cgeq=\sum_{i\in S_\cgeq} d_i$.
  \STATE Let $S_\cleq=\{i : \|\bx_i\|^2\leq 1-\delta\}$,
         $\Delta_\cleq=\sum_{i\in S_\cleq} d_i$.
  \STATE Let $S_\cneg=\{ij :  \bx_i\cdot \bx_j<-\delta\}$, $ \Delta_\cneg=\sum_{ij\in S_\cneg} |w_{ij}|$.
  \IF{$\Delta_\cgeq\geq\delta\alpha$}
  \STATE Let $\bA_{ii}=-{d_i}/{\Delta_\cgeq}$ for $i\in S_\cgeq$
    and $b=-1$. 
\STATE Return $(\bA,b)$.{\label{line:maxoracle1}}
  \ELSIF{$\Delta_\cleq\geq\delta\alpha$}
   \STATE Let $\bA_{ii}={d_i}/{\Delta_\cleq}$ for $i\in S_\cleq$
    and $b=1$. 
\STATE Return $(\bA,b)$.\label{line:maxoracle2}
  \ELSIF{$\Delta_\cneg\geq\delta\alpha$}
   \STATE Let $\bA_{ij}={w^h_{ij}}/{\Delta_\cneg}$ for $ij\in S_\cneg$
    and $b=0$. 
\STATE Return $(\bA,b)$.\label{line:maxoracle3}
 \ELSE
  \STATE Ignore all nodes in~$S_\cgeq$ and~$S_\cleq$ 
   and all edges in~$S_\cneg$. Let~$\bC'$ be the matrix that
   corresponds to the objective function of the modified graph $G'$.
  \IF{$\bC'\circ\bX<(1-4\delta)\alpha$}
   \STATE Let $\bA=\bC'/\alpha$ and $b=1-3\delta$. 
    Return $(\bA,b)$. \label{line:maxoracle4}
  \ELSE
   \STATE Round~$\bX$, and return the rounded solution.
\label{line:round}
  \ENDIF
\ENDIF
 \end{algorithmic}
 \caption{Oracle for \ref{sdp:maxag}. \label{alg:maxoracle}}
\end{algorithm}

\begin{lemma}\label{sdp1}
 Algorithm~\ref{alg:maxoracle} satisfies criterion (i) of Theorem~\ref{thm:mwmsdp}, i.e.,
for all returned $(\bA,b)$,  
$\bA \circ \bX \leq b - \delta$ and $\forall \bX' \in {\mathcal X},\bA \circ \bX' \geq b$ where ${\mathcal X}$ is the feasible space of \ref{sdp:maxag}. 
\end{lemma}

\begin{proof}
 For line~7,
 $\bA\circ\bX \leq \sum_{i\in S_\cgeq} -{d_i(1+\delta)}/{\Delta_\cgeq}
 = -1-\delta$, since $\|x_i\|^2\geq 1+\delta$ for all $i\in S_\cgeq$.
 On the other hand, for a feasible~$\bX'$, $\|x'_i\|^2=1$ for all~$i$.
 Hence $\bA\circ\bX'=\sum_{i\in S_\cgeq} {-d_i}/{\Delta_\cgeq} = -1$. This proves that the oracle is $\delta$-separating when it returns from line~7. 
For lines~10 and~13, the proof is almost identical.

 For line~17, 
we do not use the violated constraints;
 instead we use~$\bC'$ to construct~$\bA$, and show that 
 $\bC'\circ\bX'\geq (1-3\delta)\alpha$. We start from the fact
 that $\bC\circ\bX'\geq \alpha$, since $\bX'$ is feasible for~\ref{sdp:maxag}. By removing all nodes in~$S_\cgeq$,
 we remove all edges incident on the removed nodes. The total weight
 of removed edges is bounded by~$\Delta_\cgeq$, which is this case is less than~$\delta\alpha$.
 Similarly, we lose at most~$\delta\alpha$ for each of~$S_\cleq$ and~$S_\cneg$.
 Hence, the difference between $\bC'\circ\bX'$ and $\bC\circ\bX'$
 is bounded by $3\delta\alpha$, and so
$\bC'\circ\bX' \geq (1-3\delta)\alpha$ which implies $\bA\circ\bX'\geq 1-3\delta$.
Therefore we have $\delta$ separation because, $\bA\circ\bX = \bC' \circ \bX / \alpha < 1-4\delta$.
\end{proof}

\begin{lemma}\label{sdp2}
 Algorithm \ref{alg:maxoracle} satisfies criterion (ii) of Theorem~\ref{thm:mwmsdp}, i.e., 
$\rho \bD \succeq \bA - b\bD \succeq - \rho \bD$ for some $\rho=O(1/\delta)$.
\end{lemma}

\begin{proof} Since $|b| \leq 1$  it suffices to show
that for every positive semidefinite $\bY$,  $|\bA\circ\bY| = 
 \rho \bD\circ\bY$.
  For line~7,
 the proof is straightforward.
 To start, $\bA$ is a diagonal matrix where $|\bA_{ii}|={d_i}/{\Delta_\cgeq}
 \leq {d_i}/(\delta\alpha)$.
 On the other hand, $\bD_{ii}={d_i}/{2W}$, while
$\alpha\geq W/2$, so we have
 $|\bA_{ii}|=O({1}/{\delta})\bD_{ii}$ which proves that
 $|\bA\circ\bY| =  O({1}/{\delta})\bD\circ\bY$.
The proof is identical for line~10.

For lines~13 and~17, consider the decomposition of $\bY$, i.e., $\{\by_i\}$ such that $\bY_{ij} = \by_i \cdot \by_j$.
We use the fact that $\by_i\cdot \by_j\leq \|\by_i\|^2+\|\by_j\|^2$ for every pair of vectors~$\by_i$ and~$\by_j$. Therefore for $\bY_{ij}=\by_i \cdot \by_j$, we have at line~13,
\begin{align*}
 |\bA \circ \bY| &= \sum_{ij \in S_3}  \hspace{-0.1in}\frac{|w^h_{ij}|}{\Delta_3} \bY_{ij} 
 \leq  \sum_{ij \in S_3} \hspace{-0.1in} \frac{|w^h_{ij}|}{\Delta_3} (\|\by_i\|^2 + \|\by_j\|^2)  
= \frac{1}{\Delta_3} \sum_{i}  \|y_i\|^2 \sum_{j:ij \in S_3} |w^h_{ij}| \leq 
\frac{1}{\Delta_3} \sum_{i}  d_i \|y_i\|^2 \\ 
&= \frac{1}{\Delta_3} \sum_{i}  2W\bD_{ii} \bY_{ii} = \frac{2W}{\Delta_3} \bD \circ \bY\,,
\end{align*}
which implies $|\bA\circ\bY|\leq O({1}/{\delta})\bD\circ\bY$ given $\alpha \geq W/2$ and $\Delta_3 \geq \delta \alpha$.
For line~17,  let $H^+|_{G'},H^-|_{G'}$ denote $H^+,H^-$ as modified by line~15, then
\begin{align*}
\bA & \circ \bY =  \frac1\alpha \bC' \circ \bY =  \frac1{2\alpha} \sum_{ij \in H^+|_{G'}} 2 w^h_{ij} \bY_{ij}  + \frac1{2\alpha} \sum_{ij \in H^-|_{G'}} |w^h_{ij}| (\bY_{ii} + \bY_{jj} - 2\bY_{ij}) \\
& \leq \frac1{2\alpha} \sum_{ij \in G'} 2|w^h_{ij}| (\bY_{ii} + \bY_{jj})  \leq \frac1\alpha \sum_i d_i \bY_{ii} = \frac{2W}{\alpha} \bD \circ \bY
\end{align*}	
which implies that $ \bA \circ \bY = O(1) \bD \circ \bY$.
 Summarizing, Algorithm \ref{alg:maxoracle} is $O(1/\delta)$-bounded.
\end{proof}

Lemmas~\ref{sdp1} and \ref{sdp2}, in conjunction with Theorem~\ref{thm:mwmsdp} prove Theorem~\ref{thm:max}.
The update procedure~\cite{Steurer10} maintains (and defines) the candidate vector~$\bX$ implicitly.
In particular it uses matrices of dimension
$n \times d$, in which every entry is a (scaled) Gaussian random variable. The algorithm also uses a precision parameter (degree of the polynomial approximation to represent matrix exponentials)~$r$. 
Assuming that~$T_M$ is
the time for a multiplication between a 
returned~$\bA$ and some vector, the update process computes the $t$th $\bX$  in time $O(t \cdot r \cdot d \cdot T_M)$, a quadratic dependence on~$t$ in total. We will ensure that any returned~$\bA$ has at most $m'$ nonzero entries, and therefore $T_M=O(m')$.
The algorithm requires space that is sufficient to represent a linear combination of the matrices~$\bA$ which are returned in the different iterations.
We can bound
 $\rho=O(1/\delta)$, and therefore the total number of iterations is~$\tilde{O}(\delta^{-4})$. For our purposes, in \maxagree\ we will have
$d=O(\delta^{-2} \log n)$, $r=O(\log (\delta^{-1})$,
and $T_M=O(m')$, giving us a $\tilde{O}(n\delta^{-10})$ time and $\tilde{O}(n\delta^{-2})$ space algorithm.
However, unlike the general~$\bX$ used in Steurer's approach, in our oracle
the~$\bX$ is used in a very specific way.
This leaves open the question of determining the exact space-versus-running-time tradeoff.

\paragraph{Rounding the Fractional Solution:}
Note that the solution of the SDP found above is only approximately
feasible.
Since the known rounding algorithms can not be applied in a black box fashion, the following Lemma proves the rounding algorithm.

\begin{lemma}\label{hypround}
 If Algorithm~\ref{alg:maxoracle} returns a clustering solution,
 it has
 at least $0.7666(1-O(\delta))\alpha$ agreements.
\end{lemma}

\begin{proof}
 We show that the rounding algorithm returns a clustering with
 at least $0.7666(1-O(\delta))\bC'\circ\bX$ agreements. Combined with the
 fact that $\bC'\circ\bX>(1-4\delta)\alpha$ (line~\ref{line:round}),
we obtain the desired result.

 Since we deal with~$\bC'$ instead of~$\bC$, we can ignore all nodes
 and edges in~$S_\cgeq$, $S_\cleq$, and~$S_\cneg$.
 We first rescale the vectors in~$\bX$ to be unit vectors. Since all
 vectors that are not ignored (not in~$S_\cgeq$ nor~$S_\cleq$)
have length between $1-O(\delta)$ and $1+O(\delta)$ (since we take the square root),
 this only changes the objective value by $O(\delta w_{ij})$ for each edge. Hence the total decrease is bounded by $O(\delta W)=O(\delta\alpha)$.

 We then (1) first change the objective value of edges $(i,j)$ with
 $-\delta<x_i\cdot x_j<0$ by ignoring them, and only then (2) consider fixing the violated constraints $x_{i}\cdot x_{j} <0$ to produce a feasible or integral solution.

Step (1)  decreases the objective value by at most~$\delta |w_{ij}|$ for each negative edge. Again, the objective value
 decreases by at most $O(\delta\alpha)$. For step (2) we use Swamy's rounding algorithm~\cite{Swamy04}, which
 obtains a~$0.7666$ approximation factor.
The constraint $x_i\cdot x_j\geq 0$ required by Swamy's algorithm
is not satisfied for some edges.
 However, the rounding algorithm is based on random hyperplanes
 and the probability that~$x_i$ and~$x_j$ are split by a hyperplane
 only increases as $x_i\cdot x_j$ decreases. For positive edges,
 we already accounted for this in step (1) when the value of the edge was made~$0$.
 For negative edges, the probability that~$i$ and~$j$ land in different
 clusters only increases by having negative $x_i\cdot x_j$, but again, the contribution to the objective is still~$0$. 
 Therefore, we obtain a clustering that has at least
 $0.7666(1-O(\delta))\bC'\circ\bX$ agreements.
\end{proof}

\section{Multipass Algorithms}
\label{sec:multipass}

In this section, we present  $O(\log \log n)$-pass algorithms
 for \mindisagree~on unit weight graphs:
these apply to both a fixed and unrestricted number of clusters.
In each pass over the data, the algorithm is presented with the
same input, although not necessarily in the  same order. 

\subsection{$\mindisagree$ with Unit Weights}
\label{sec:mind3}

Consider the $3$-approximation algorithm for $\mindisagree$ on unit-weight graphs due to \citet{AilonCN08}.

 \begin{algorithmic}[1]
\STATE Let $v_1, \ldots , v_n$ be a uniformly random ordering of~$V$.
  Let~$U\leftarrow V$ be the set of ``uncovered'' nodes. 
\FOR{$i=1$ to $n$}
\IF{$v_i\in U$} 
\STATE  Define $C_i\leftarrow\{v_i\} \cup \{v_j\in U: v_i v_j\in E^+\}$ and let  $U\leftarrow U\setminus C_i$. We say $v_i$ is ``chosen''.
\ELSE
\STATE $C_i \leftarrow \emptyset$.
\ENDIF
\ENDFOR
\STATE {\bf Return} the collection of non-empty sets $C_i$.
\end{algorithmic}

It may appear that emulating the above algorithm in the data stream
model requires~$\Omega(n)$ passes, since determining whether~$v_i$ should
be chosen may depend on whether~$v_j$ is chosen for each $j<i$.
However, we will show that  $O(\log \log n)$-passes suffice.
This improves upon a result by \citet{ChierichettiDK14},
who developed a modification of the algorithm that used $O(\epsilon^{-1}
\log^2 n)$ streaming passes and returned a $(3+\epsilon)$-approximation,
rather than a $3$-approximation.
Our improvement is based on the following lemma:

\begin{lemma}\label{lem:smallF}
Let $U_t$ be the set of uncovered nodes after iteration~$t$ of the
above algorithm,
and let \[F_{t,t'}=\{v_i v_j \in E^+, i,j\in U_t, t < i,j\leq t'\} \ .\]
With high probability, $|F_{t,t'}| \leq 5 \cdot \ln n \cdot t'^2/t$. 
\end{lemma}

\begin{proof}
Note that the bound holds vacuously for $t\leq 10\log n$ so in the rest of the proof we will assume $t\geq 10\log n$. Fix the set of $t'$ elements in the random permutation and consider the induced graph $H$ on these $t'$ elements. Pick an arbitrary node $v$ in $H$. We will consider the random process that picks each of the first $t$ entries of the random permutation by picking a node in $H$ uniformly at random without replacement. We will argue that at the end of these $t$ steps, with probability at least $1-1/n^{10}$, either $v$ is covered or at most $\alpha t'/t$ neighbors of $v$ in $H$ are uncovered where $\alpha=10\log n$. Hence, by the union bound, all uncovered nodes have at most $\alpha t'/t$ uncovered neighbors and hence the number of edges in $H$ whose both endpoints are uncovered after the first $t$ steps is at most $(\alpha t'/t) \cdot t'/2$. The lemma follows because $F_{t,t'}$ is exactly the number of edges in $H$ whose both endpoints are uncovered after the first $t$ steps.

To show that after $t$ steps, either $v$ is covered or it has at most $\alpha t'/t$ uncovered neighbors we proceed as follows. Let $B_i$ be the event that after the $i$th iteration, $v$ is not covered and it has at least $\alpha t'/t$ uncovered neighbors. Then, since $B_{i+1}\subset B_i$ for each $i$,
\begin{eqnarray*}
 \prob{v \mbox{ is covered or it has at most } \alpha t'/t \mbox{ uncovered neighbors.}}
& =& 1-\prob{B_r} \\
& = & 1-\prob{B_r \cap B_{r-1} \cap \ldots \cap B_1}\\
&=& 1-p_r p_{r-1} \ldots p_1
\end{eqnarray*}
where $p_i=\prob{B_i| B_1 \cap B_2 \cap \ldots \cap B_{i-1}}$.
Note that 
\[p_i\leq 1-\prob{v \mbox{ gets covered at step } i | B_1 \cap B_2 \cap \ldots \cap B_{i-1}}\leq 1-\frac{\alpha t'/t+1}{t'-(i-1)} <1-\alpha/t\ ,\] 
and hence, 
\[\prob{v \mbox{ is covered or it has at most } \alpha t'/t \mbox{ uncovered neighbors.}}
\geq  1-(1-\alpha/t)^t\geq 1-\exp(-\alpha)=1-1/n^{10}
\]
as required.

\end{proof}

\noindent{\bf Semi-Streaming Algorithm.} 
As a warm-up, first consider the following two-pass streaming algorithm that emulates Ailon et al.'s algorithm using $O(n^{1.5}\log^2 n)$ space:
\begin{enumerate}
\item \emph{First pass:} Collect all edges in $E^+$ incident on $\{v_i\}_{i\in [\sqrt{n}]}$. This allows us to simulate the first $\sqrt{n}$ iterations of the algorithm.
\item \emph{Second pass:} Collect all edges in $F_{\sqrt{n},n}$.
 This allows us to simulate the remaining $n-\sqrt{n}$ iterations.
\end{enumerate}
The space bound follows since  each pass requires storing only $O({n}^{1.5} \log n)$ edges with high probability.
 requires storing at most $n^{1.5}$ edges and, with high probability, the second pass requires storing $|F_{\sqrt{n},n}|=O({n}^{1.5} \log n)$ edges.

Our semi-streaming algorithm proceeds  as follows.
\begin{itemize}
\item For $j\geq 1$, let $t_j=(2n)^{1-1/2^{j}}$: during the $(2j-1)$-th pass,
we store all edges in $F_{t_{j-1},t_j}$ where $t_0=0$, and during the $(2j)$-th pass we determine~$U_{t_j}$.
\item After the $(2j)$-th pass we have simulated the first~$t_j$ iterations of
\citeauthor{AilonCN08}'s algorithm. Since $t_j\geq n$ for $j=1+\log \log n$, our algorithm terminates after $O(\log \log n)$ passes. 
\end{itemize}

\begin{theorem}
\label{thm:mult1}
  On a unit-weight graph,
  there exists a $O(\log \log n)$-pass 
  semi-streaming algorithm that returns with high probability a $3$-approximation to \mindisagree.
\end{theorem}

\begin{proof}
In the first pass, we need to store  at most $t_1^2=((2n)^{1-1/2})^2=2n$ edges.
For the odd numbered passes after the first pass, by Lemma \ref{lem:smallF}, the space is at
most 
\[5 \cdot \ln n \cdot t_j^2/t_{j-1}
=5 \cdot \ln n \cdot
(2n)^{2-2/2^{j}}/(2n)^{1-1/2^{j-1}}
=5 \cdot \ln n \cdot 2n 
= O(n \log n)\,,\] with high
probability.
The additional space used in the even numbered passes is trivially bounded by
$O(n \log n)$.
The approximation factor follows from the analysis of~\cite{AilonCN08}.
\end{proof}

\subsection{$\mindisagree_k$ with Unit Weights}
\label{sec:mdfc}

Our result in this section is based the following algorithm of~\cite{GiotisG06} that returns a $(1+\epsilon)$-approximation for $\mindisagree_k$ on unit-weight graphs. Their algorithm is as follows:

\begin{enumerate}
\item Sample $r=\poly(1/\epsilon,k)\cdot \log n$ nodes $S$ and for every possible $k$-partition $\{S_i\}_{i\in [k]}$  of $S$: 
\begin{enumerate}
\item Compute the cost of the clustering where $v\in V\setminus S$ is assigned to the $i$th cluster where
\[
i=\argmax_{j}\left ( \sum_{s\in S_j: sv\in E^+} w_{sv}+ \sum_{s\not \in S_j: sv\in E^-} |w_{sv}| \right )
\]
\end{enumerate}
\item Let $\calC'$ be the best clustering found. If all clusters in~$\calC'$ have at least $n/(2k)$ nodes, return~$\calC'$. Otherwise, fix all the clusters of size at least $n/(2k)$ and recurse (with the appropriate number of centers still to be determined) on the set of nodes in clusters that are smaller than $n/(2k)$. 
\end{enumerate}

We first observe the above algorithm can be emulated in $\min(k-1,\log n)$ passes in the data stream model. To emulate each recursive step in one pass we simply choose $S$ are the start of the stream and then collect  all incident edges on~$S$. We then use the $\disagree$ oracle developed in Section~\ref{sec:l2combrec} to find the best possible partitions during post-processing. It is not hard to argue that this algorithm terminates in $O(\log n)$
rounds, independent of~$k$: Call clusters with fewer than $n/2k$ nodes ``small'', and those with
at least $n/2k$ nodes ``large''. 
Observe that the number of nodes in small clusters halves in each round since there are at most $k-1$ small clusters and each has at most $n/(2k)$ nodes. This would suggest a $\min(k-1,\log n)$ pass data stream algorithm, one pass to emulate each round of the offline algorithm. However, the next theorem shows that the algorithm can actually be emulated in $\min (k-1, \log \log n)$ passes.

\begin{theorem}\label{thm:mult2}
There exists a $\min (k-1, \log \log n)$-pass $O(\poly(k,\log n,1/\epsilon) n)$-space algorithm that $(1+\epsilon)$ approximates $\mindisagree_k(G)$.
\end{theorem}
\begin{proof}
To design an $O(\log \log n)$ pass algorithm, we proceed as follows.
At the start of the $i$-th pass, suppose we have $k'$ clusters still to determine
  and that~$V_{i}$ is the set of remaining nodes that have not
  yet been included in large clusters.
We will pick $k'$ random sets of samples
    $S_1, \ldots, S_{k'}$ in parallel from~$V_{i}$ each of size 
    \[N_i = 2 rn^{2^{i-1}/\log n} \ .\] 

For each sampled node, we extract all edges to unclustered
nodes.  We will use this information to emulate one or more rounds of the algorithm. Note that  since $N_i\geq n$ for $i\geq 1+\log \log n$, the algorithm must terminate in $O(\log \log n)$ passes since  in pass $1+\log \log n$ we are storing all edges in the unclustered graph. What remains is to establish a bound on the space required in each of the passes. To do this we will first argue that in each pass, the  number of unclustered nodes drops significantly, perhaps to zero.

Since there are only $k'$ clusters still to determine, and every round
of the algorithm fixes at least one cluster, it is conceivable that
the sets $S_1, \ldots, S_{k'}$ could each be used  to emulate one of
the remaining $\leq k'$ rounds of the algorithm; this would suggest it
is possible to completely emulate the algorithm in a single pass.
However, this will not be possible if at some point there are fewer than $r$ unclustered nodes remaining in all the sets $S_1, \ldots, S_{k'}$. At this point, we terminate the current set of samples, and take a new pass. Observe that in this case we have likely made progress, as the number of
unclustered nodes over which we are working has likely dropped significantly. Specifically, suppose the number of unclustered nodes is greater than $|V_i|n^{2^{i-1}/\log n}$ before we attempt to use $S_{k'}$. By the principle of deferred decision, the expected number of unclustered nodes in $S_{k'}$ is at least 
\[\frac{ |V_i|n^{2^{i-1}/\log n}}{|V_i|}  N_i  = 2 r \ .\]
Therefore, by an application of the Chernoff bound, we can deduce the number of unclustered nodes when we terminate the current pass is less than $|V_i|n^{2^{i-1}/\log n}$, i.e., the number of unclustered nodes has decreased by a factor of at least $n^{2^{i-1}/\log n}$ since the start of the  pass. 

Applying this analysis to all passes and using the fact that $|V_1|=n$, we conclude that 
\[
|V_{i+1}|\leq \frac{|V_{i}|}{n^{2^{i-1}/\log n}} \leq 
  \frac{|V_1|}{n^{2^{1-1}/\log n}\cdot n^{2^{2-1}/\log n} \cdot \ldots \cdot n^{2^{i-1}/\log n}}
= \frac{n}{n^{(2^{i}-1)/\log n}} \ . 
\]
The space needed by our algorithm for round~$i$ is therefore
$O(|V_i| N_i k') = O(krn^{1+1/\log n})=\tilde{O}(krn)$.  
\end{proof}

\section{Lower Bounds}
\label{sec:lower}
Finally, we consider the extent to which our results can (not) be
improved, by showing lower bounds for variants of problems that we can
solve.  All our proofs will use the standard technique of reducing from two-party communication complexity problems, i.e., Alice has input $x$ and Bob has input $y$ and they wish to compute some function $f(x,y)$ such that the number of bits communicated between Alice and Bob is small. A lower bound on the number of bits communicated can be used to lower bound  the space complexity of a data stream algorithm as follows. Suppose Alice can transform $x$ in to the first part $S_1$ of a data stream and Bob can transform $y$ in to the second part $S_2$ such that the result of the data stream computation on $S_1\circ S_2$ implies the value of $f(x,y)$. Then if the data stream algorithm takes $p$ passes and uses $s$ space, this algorithm can be emulated by Alice and Bob using $2p-1$ messages each of size bits $s$; Alice starts running the data stream algorithm on $S_1$ and each time a player  no longer has the necessary information to emulate the data stream algorithm they send the current memory state of the algorithm to the other player. Hence, a lower bound for the communication complexity problem yields a lower bound for the data stream problem.

\begin{theorem}\label{thm:n2lb}
A one-pass stream algorithm that tests whether $\mindisagree(G)=0$,
with probability at least~$9/10$,
requires~$\Omega(n^2)$ bits if permitted weights are $\{-1,0,1\}$. 
\end{theorem}

\begin{proof}
The theorem follows from a reduction from the communication problem \ind. Alice has a string $x\in \{0,1\}^{{n\choose 2}}$, indexed as $[n]\times [n]$ and unknown to Bob, and Bob wants to learn $x_{i,j}$ for some $i,j\in [n]$ that is unknown to Alice. 
Any one-way protocol from Alice to Bob that allows Bob to learn $x_{i,j}$ requires 
$\Omega(n^2)$ bits of communication~\citep{Ablayev96}. 

Consider the protocol for \ind\ where Alice creates a graph~$G$ over
nodes $V = \{v_1,\ldots, v_n\}$ and adds edges
$\{\{v_i,v_j\}:x_{i,j}=1\}$ each with weight~$-1$. She runs a data
stream algorithm on $G$ and sends the state of the algorithm to Bob
who adds positive edges $\{u,v_i\}$ and $\{u,v_j\}$ where~$u$ is a new
node.
All edges without a specified weight are treated as not present, or
equivalently as having weight zero. 
Hence the set of weights used in this graph is  $\{-1, 0, +1\}$. 
Now, if~$x_{ij} = 0$, then  $\disagree(G)=0$: consider
the partition containing $\{u,v_i,v_j\}$, with each other item comprising
a singleton cluster.
Alternatively, $x_{ij}=1$ implies $\disagree(G)\geq1$ since a clustering must disagree with one of the three edges on $\{u,v_i,v_j\}$. It follows that  every data stream algorithm returning a multiplicative estimate of $\mindisagree(G)$ requires~$\Omega(n^2)$ space. 
\end{proof}

When permitted weights are restricted to $\{-1,1\}$, the following multi-pass lower bounds holds:

\begin{theorem}\label{thm:n3lb}
A $p$-pass stream algorithm that tests whether $\mindisagree(G)=0$,
with probability at least~$9/10$,
requires~$\Omega(n/p)$ bits when permitted weights are $\{-1,1\}$.
\end{theorem}
\begin{proof}
The proof uses a reduction from the communication problem of \disj\ where Alice and Bob have strings $x,y\in \{0,1\}^{n}$ and wish to determine where there exists an $i$ such that $x_i=y_i=1$. 
Any $p$ round protocol between  Alice and Bob requires~$\Omega(n)$ bits of
communication~\citep{KalyanasundaramS92} and hence there must be a message of $\Omega(n/p)$ bits. 

Consider the protocol for \disj\ on a graph~$G$ with nodes $V = \{a_1,\ldots, a_n, b_1, \ldots, b_n, c_1,\ldots, c_n\}$.
For each $i\in [n]$, Alice adds an edge $\{a_i,b_i\}$ with weight $(-1)^{x_i+1}$. She runs a data stream algorithm on $G$ and sends the state of the algorithm to Bob. For each $i\in [n]$, Bob adds an edge $\{b_i,c_i\}$ of weight $(-1)^{y_i+1}$ along with negative edges 
\[\{\{a_i,c_i\}:i\in [n]\} \cup \{\{u,v\}:u\in \{a_i,b_i,c_i\}, v\in \{a_j,b_j,c_j\}, i\neq j\} \ .\] Note that $\mindisagree(G)>0$ iff there exists $i$ with $x_i=y_i=1$. Were there no such~$i$, the positive edges would all be isolated, whereas if $x_i=y_i=1$ then every partition violates one of the edges on $\{a_i,b_i,c_i\}$. It follows that every $p$-pass data stream algorithm returning a multiplicative estimate of $\mindisagree(G)$ requires~$\Omega(n/p)$ space. 
\end{proof}

Next we show a lower bound that applies when the number of negative weight edges in bounded. This shows that our upper bound in Theorem \ref{thm:min} is essentially tight.

\begin{theorem}\label{thm:neednegatives}
A one-pass stream algorithm that tests whether $\mindisagree(G)=0$,
with probability at least~$9/10$,
requires~$\Omega(n+|E^{-}|)$ bits if permitted weights are $\{-1,0,1\}$. \end{theorem}
\begin{proof}
A lower bound of $\Omega(|E^-|)$ follows by considering the construction in Theorem \ref{thm:n2lb} on $\sqrt{|E^-|}$ nodes. A lower bound of $\Omega(n)$ when $n\geq |E^{-}|$ follows by considering the construction in Theorem \ref{thm:n3lb} without adding the negative edges $\{uv:u\in \{a_i,b_i,c_i\}, v\in \{a_j,b_j,c_j\}, i\neq j\}$.
\end{proof}

Finally, we show that the data structure for evaluating 2-clusterings
of arbitrarily weighted graphs (Section~\ref{sec:3ds}) cannot be
extended to clusterings with more clusters. 

\begin{theorem}\label{lem:3clusterquery}
When $|\calC|= 3$, a data structure that returns a multiplicative estimate of $\disagree(G,\calC)$ with probability at least~$9/10$,
requires $\Omega(n^2)$ space.
\end{theorem}

\begin{proof}
We show a reduction from the communication problem of \ind\ where Alice has a string $x\in \{0,1\}^{n^2}$ indexed as $[n]\times [n]$ and Bob wants to learn $x_{i,j}$ for some $i,j\in [n]$ that is unknown to Alice. 
A one-way protocol from Alice to Bob that allows Bob to learn~$x_{i,j}$ requires~$\Omega(n^2)$ bits of communication~\cite{Ablayev96}. 
Consider the protocol for \ind\ where Alice creates a graph $G$ over nodes $V = \{a_1,\ldots, a_n,b_1, \ldots, b_n\}$ and adds edges $\{a_u b_v:x_{u,v}=1\}$ each with weight~$-1$. She runs a data stream algorithm on $G$ and sends the state of the algorithm to Bob who then queries the partition $\C = \{ a_i b_j, \{a_\ell :\ell\neq i\}, \{b_\ell :\ell\neq j\}\}$. Since $\disagree(G,\C)=x_{ij}$ it follows that  every data stream algorithm returning multiplicative estimate of $\disagree(G,\C)$ requires~ $\Omega(n^2)$ space. 
\end{proof}

\bibliographystyle{plainnat} \bibliography{cc}

\appendix
\section{Extension to Bounded Weights}
\label{sec:ggext}
In this section, we detail the simple changes that are required in the paper by \citet{GiotisG06} such that their result extends to the case
 where there are no zero weights and the magnitude of all
 non-zero weights is bounded between 1 and $w_*$ where we
 will treat $w_*$ as constant.

\paragraph{Max-Agreement.} See Section \ref{sec:2ds} for a description of the max-agreement algorithm.
The proof in the unweighted case first shows a lower bound for $\maxagree_k(G)$ of 
\[\max(|E^+|,|E^-|(1-1/k))\geq n^2/16 \,.\]
In the bounded-weights case, the magnitude of every edge only increases and so the same bound holds.
Hence, for the purpose of returning a $(1+O(\epsilon))$ multiplicative
approximation, it still suffices to find an~$\epsilon n^2$
additive approximation.
Indeed, the argument of \citeauthor{GiotisG06} still applies, with small changes
by decreasing~$\epsilon$ by a factor~$w_*$ and increasing~$r$ by a factor
of~$w_*^2$.
Rather than retread the full analysis of~\citet{GiotisG06},
we just identify the places where their argument is altered. 

The central result needed is that
estimating the cost associated with
placing each node in a given cluster can be done accurately from a
sample of the clustered nodes.
This is proved via a standard additive
Chernoff bound (Lemma~3.3 of \citet{GiotisG06}).
It is natural to define the weighted generalization of this
estimate based on the weights of edges in the sample and to rescale
accordingly. 
One can then apply the additive Chernoff bound over random variables which
are constrained to have magnitude in the range $\{1,2,\ldots,w_*\}$,
rather than $\{0,1\}$ as in the unit-weights case.
The number of nodes whose estimated relative contribution deviates by more than 
$(\epsilon/32w_*)$ from its (actual) contribution to the optimal clustering is
then bounded by applying the Markov inequality. 
Provided we increase the sample size~$r$ by a factor of~$w_*^2$,
these bounds all hold with the necessary probability.

The other steps in the argument are modified in a similar way: we analyze
the total weight of edges in agreement, rather than their number.
Specifically, applying this modification to 
Lemma~3.4 of \cite{GiotisG06}, we bound the impact of misplacing one
node in the constructed clustering compared to the optimal
clustering. 
With the inequality from the above Chernoff bound
argument, the impact of this can, as in the orignal argument, be
bounded in the weighted case
by~$(\epsilon/8)n$.
The number of nodes for which this does not hold is at most
a fraction $(\epsilon/8w_*)$ of each partition, and so contribute to a
loss of at most $(\epsilon^2/8)n^2$ (weighted) agreements in each
step of the argument, as in the original analysis. 

\paragraph{Min-Agreement.} See Section \ref{sec:mdfc} for a description of the min-agreement algorithm.
Again, the central step is the use of a Chernoff bound on edges
incident on sampled nodes.
Modifying this to allow for bounded-weight edges again incurs a factor
of $w_*^2$, but is otherwise straightforward.
It then remains to follow through the steps of the original argument,
switching from cardinalities of edgesets to their weights.  

\end{document}